\documentclass[11pt]{article}
\usepackage{amsthm}
\usepackage{amsmath}
\usepackage{amssymb}
\usepackage{color}
\usepackage{makecell}
\usepackage{array}
\usepackage{diagbox}
\usepackage{ifpdf}
\usepackage{xcolor}
\usepackage{txfonts}
\usepackage{animate}
\usepackage[hidelinks]{hyperref}
\usepackage{url}
\usepackage[T1]{fontenc}
\usepackage[all]{xy}
\usepackage{epic,bez123}
\ifpdf 
  \usepackage[pdftex]{graphicx}
  \DeclareGraphicsExtensions{.pdf,.png,.jpg,.jpeg,.mps}
  \usepackage{pgf}
  \usepackage{tikz}
\else 
  \usepackage{graphicx}
  \DeclareGraphicsExtensions{.eps,.bmp}
  \usepackage{pgf}
  \usepackage{tikz}
\fi
\usepackage{epic,bez123}
\usepackage{wrapfig}
\usepackage{pgfplots}
\newtheorem{theorem}{Theorem}[section]
\newtheorem{remark}[theorem]{Remark}

\newtheorem{lemma}[theorem]{Lemma}
\newtheorem{proposition}[theorem]{Proposition}

\newtheorem{definition}[theorem]{Definition}

\newtheorem{corollary}[theorem]{Corollary}

\newtheorem{example}[theorem]{Example}

\def\urltilda{\kern -.15em\lower .7ex\hbox{\~{}}\kern .04em}
\date{} %
\def\keywords#1{\begin{center}{\bf Keywords}\\{#1}\end{center}} %

\newcommand{\F}{{\mathbb F}}
\newcommand{\var}{{\rm Var}}
\newcommand{\w}{{\rm w_H}}
\newcommand{\nl}{{\rm nl}}
\newcommand{\nnl}{{\rm nnl}}
\newcommand{\rnnl}{{\rm Rnnl}}

\newcommand{\NW}{{\rm NW}}
\newcommand{\supp}{{\rm supp}}
\newcommand{\codim}{{\rm codim}}

 \usepackage[textsize=footnotesize]{todonotes}
\usepackage{lipsum}
\usepackage{marginnote}

\title {\huge\bf Higher-order nonhomomorphicity under affine restrictions}
\author{\Large {\bf Serge Feukoua, Ana S\u{a}l\u{a}gean and Pantelimon St\u{a}nic\u{a} }
\vspace{.3cm} \\ 
Department of Fundamental Sciences,  National Advanced School of Public Works, \\
Yaound\'e 610101, Cameroon\\
Department of Computer Science, Loughborough University, \\Loughborough 10587, UK\\
Department of Applied Mathematics, Naval Postgraduate School,\\
Monterey, CA 93943, USA}
\date{\today}

\begin{document}

\maketitle

\abstract{We study the behavior of the $k$-th-order nonhomomorphicity of Boolean functions under restriction to affine subspaces. For odd $k\ge3$ and $0\le t<n$, we introduce the $t$-restriction $k$-th-order nonhomomorphicity $P_{k,t}(f)$, defined as the minimum of $P_k(f_{\mid A})$ over affine spaces of codimension at most $t$. For a direct sum of $m$ monomials of degrees at least two, we determine $P_{k,t}$ exactly. When $t<m$, a minimum is obtained by setting one variable equal to zero in each of the $t$ monomials of lowest degree. When $m\le t<n$, the minimum is zero. The proof combines multiplicativity of even Walsh moments with a contraction argument for affine restrictions. We establish a general formula for the Walsh transform of a Boolean function restricted to an affine subspace of arbitrary codimension, yielding an exact spectral characterization of the $t$-restriction normalized nonlinearity $\rnnl_t(f)$. As odd $k$ tends to infinity, $(1-2P_{k,t}(f))^{1/(k+1)}$ converges to $1-2\rnnl_t(f)$. We further prove that every non-affine function has a hyperplane restriction with strictly smaller odd-order nonhomomorphicity. Finally, we determine $\rnnl_1(T_{d,n})$ exactly for all threshold functions with $n\ge3$. This proves the two threshold-function conjectures left open in~\cite{SergeAnaPante2025}, including the previously unresolved odd-majority case, where a minimizing hyperplane is $x_1+x_2=0$.}


\section{Introduction}
\label{sec:introduction}
A Boolean function $f$ in $n$ variables is a mapping from the set
$\{0,1\}^n$ to $\{0,1\}$.
Boolean functions are used in symmetric cryptography and in algebraic coding theory, including the study of Reed--Muller and Kerdock codes~\cite{rueppel2012analysis,Carlet2021}. Nonhomomorphicity tests estimate Walsh moments by evaluating the function on randomly sampled tuples, without computing its whole spectrum. This is useful when the number of variables makes a full Walsh transform impractical.

In cryptography, nonlinearity, algebraic degree, and algebraic immunity describe different aspects of resistance to attacks. An attacker may constrain the input to an affine subspace, so it is useful to determine how such restrictions change these parameters.

Carlet et al.~\cite{CSA2:2024} studied the stability of algebraic degree under affine restrictions. Feukoua et al.~\cite{SergeAnaPante2025} studied restriction nonlinearity for direct sums of monomials and threshold functions. Here we study odd-order nonhomomorphicity under affine restrictions and resolve the two threshold-function conjectures left open in that work.

For odd $k\ge3$ and $0\le t<n$, define the $t$-restriction $k$-th-order nonhomomorphicity $P_{k,t}(f)$ as the minimum of $P_k(f_{\mid A})$ over affine subspaces $A$ of codimension at most $t$. Our first main result determines this parameter for direct sums of $m$ monomials of degrees at least two. When $t<m$, the minimum is attained by setting one variable equal to zero in each of the $t$ monomials of lowest degree. When $m\le t<n$, setting one variable equal to zero in every monomial gives the minimum zero. The proof uses the factorization of even Walsh moments over direct sums and their contraction under the affine coupling produced by the pivot construction. These estimates give a lower bound for every admissible affine restriction, and the coordinate restrictions attain it.

Our second contribution concerns the $t$-restriction normalized
nonlinearity $\rnnl_t(f)$, defined as the minimum normalized
nonlinearity over all affine restrictions of codimension at most $t$.
We establish a general formula for the Walsh transform of a Boolean function restricted to an affine subspace of arbitrary codimension. This yields an exact spectral characterization of $\rnnl_t(f)$, proves that
\[
\lim_{\substack{k\to\infty\\ k\ {\rm odd}}}
\bigl(1-2P_{k,t}(f)\bigr)^{1/(k+1)}
=
1-2\rnnl_t(f),
\]
and gives an exact formula for the case of codimension one in terms of
the two largest absolute normalized Walsh coefficients, counted at distinct frequencies. We then
determine $\rnnl_1(T_{d,n})$ for every threshold function $T_{d,n}$ with $n\ge3$.
The resulting formula proves Conjectures~1 and~2 from
\cite{SergeAnaPante2025}. In particular, for odd majority, the minimum
is attained on the hyperplane $x_1+x_2=0$.

We also prove that the restriction values $P_{k,t}(f)$ decrease strictly until they reach zero. The proof uses an exact average of the moments on two parallel hyperplanes.

The paper is organized as follows. Section~\ref{sec:preliminaries} recalls the required notions and basic properties of Walsh transforms, nonlinearity, and $k$-th-order nonhomomorphicity. Section~\ref{sec:direct-sums-restrictions} studies direct sums of monomials and proves the exact formula for $P_{k,t}$. Section~\ref{sec:threshold-functions} develops the theory of restriction normalized nonlinearity, establishes the general affine-restriction theorem and its spectral consequences, and applies these results to threshold functions. Section~\ref{sec:conclusion} summarizes the main conclusions.

\section{Preliminaries}
\label{sec:preliminaries}

We let $\mathbb{F}_2$ be the finite field with two
elements and $\mathbb{F}_2^n$ be the vector space over
$\mathbb{F}_2$ of dimension~$n$.   A function $f\colon \mathbb{F}_2^n \rightarrow \mathbb{F}_2$ is an
$n$-variable Boolean function. It can be uniquely represented in {\em Algebraic Normal Form} (ANF), that is, as a polynomial function of degree at most one in each variable,  $f(x_1,\dots ,x_n)=\sum_{I\subseteq \{1,\dots
,n\}}a_I\prod_{i\in I}x_i$, 
with $a_I\in \mathbb{F}_2$ (see for example \cite{Carlet2021,cusickStanica2017}).

 When computing $k$-th-order nonhomomorphicity, it is convenient to view functions up to a suitable equivalence relation, namely affine equivalence, which we will recall below. First recall that an affine automorphism of $\mathbb{F}_2^n$ is a mapping $\varphi$
 of the form $\varphi(x) = Mx+a$,  where
$x$, $a$ and $\varphi(x)$ are viewed as column vectors, $x^T=(x_1,\ldots, x_n)$,
$a^T=(a_1,\ldots, a_n) \in \F_2^n$ and $M$ is an $n\times n$
invertible matrix over $\F_2$.
We use the following definitions.
\begin{definition}
\begin{enumerate}
\item[$(i)$] Two Boolean functions $f$ and $g$ in $n$ variables are said to be affinely equivalent and we write then $f\sim g$, if there  exists  an affine automorphism $L$ of $\mathbb{F}_2^n$, such that $f=g\circ L$, where '$\circ$' is the composition.
\item[$(ii)$] Two Boolean functions $f$ and $g$ in $n$ variables are said to be EA-equivalent, and we write $f\sim_{\mathrm{EA}}g$, 
if there  exists  an affine automorphism $L$ of $\mathbb{F}_2^n$, an affine Boolean function $\ell$  in $n$ variables, such that $f=g\circ L+\ell$.
\end{enumerate}
\end{definition}
Recall that the Walsh-Hadamard transform of a Boolean function
$f$ is the function $W_f\colon \F_2^n \rightarrow \mathbb{R}$ defined as
$\displaystyle W_f(u)=\sum_{\substack{x\in \mathbb{F}_2^n}}(-1)^{f(x)+u\cdot x},$
 for all $u\in\mathbb{F}_2^n$, where ``$\cdot$ ''  {denotes
the usual inner product}. 
The nonlinearity $\nl(f)$ of
a Boolean function $f$ over $\mathbb{F}_2^n$ is the minimum Hamming
distance $d_H(f,h)=|\{x\in \mathbb{F}_2^n\mid f(x)\neq h(x)\}|$ between
$f$ and affine functions $h$.  
We have
$\displaystyle
\nl(f)=2^{n-1}-\frac{1}{2} \max_{\substack{a\in\mathbb{F}_2^n}}\left|W_f(a)\right|$ (see, for example~\cite{cusickStanica2017}).

For our purpose it is useful to use  normalized versions of the Walsh-Hadamard transform and nonlinearity.
\begin{definition}\label{def:nnl}
  The normalized Walsh-Hadamard transform of a Boolean function $f$ in $n$ variables, denoted here  by $\NW_f$, is the function $\NW_f\colon \F_2^n \rightarrow \mathbb{R}$ defined as  $\displaystyle 
    \NW_f(a)=\frac{1}{2^n}W_f(a),
$
for all $a\in\F_2^n$.
The normalized nonlinearity of $f$ is defined by $\displaystyle 
\nnl(f)= 
\frac{1}{2^n}\nl(f)=\frac{1}{2}\left(1-\max_{a\in\F_2^n}\left|\NW_f(a)\right|\right).
$
\end{definition}
\begin{definition} Let $f:\, \F_2^n\longrightarrow\F_2$ and $k\geq 2$ be an integer.
The  $k$-th-order nonhomomorphicity of a function $f$ is the
probability $P_k(f)$ of failing the following test.
$$f(u_1+\cdots+u_{k})+f(u_1)+\cdots+f(u_{k})=0$$
 with the probability taken over all tuples $(u_1,\ldots u_{k})\in (\F_2^n)^{k}$ that is, $$P_k(f)=\dfrac{|\{(u_1,\ldots u_{k})\in \F_2^{nk}\mid \quad f(u_1+\cdots+u_{k})+f(u_1)+\cdots+f(u_{k})\neq0\}|}{2^{nk}}$$
\end{definition}
From \cite{DSZU06}, we have \begin{equation}\label{firstformulaof pkf}P_k(f)=\dfrac{1}{2}\left(1-\sum_{a\in\F_2^n}\NW_f(a)^{k+1}\right),\end{equation}
where $\NW_f$ is the normalized Walsh transform.

For a real number $q\ge2$, write
$M_q(f)=\sum_u|\NW_f(u)|^q$.
Character orthogonality gives Parseval's identity and Walsh inversion,
\[
\sum_u\NW_f(u)^2=1,\qquad (-1)^{f(x)}=\sum_u\NW_f(u)(-1)^{u\cdot x}.
\]
Since $|\NW_f(u)|\le1$, it follows that $0<M_q(f)\le1$. For odd $k\ge3$, equation~\eqref{firstformulaof pkf} is $P_k(f)=\frac12(1-M_{k+1}(f))$.

\begin{lemma}\label{lem: EA equiv-nonhomomor}
Let $k\ge3$ be odd. Then $P_k$ is invariant under EA-equivalence. This need not hold for even $k$.
\end{lemma}
\begin{proof}
Suppose $g(x)=f(Mx+v)+a\cdot x+b$, where $M$ is invertible. Changing variables from $x$ to $y=Mx+v$ gives
\begin{align}\label{eq-EAequiv-nhommo}
W_g(c)&=\sum_x(-1)^{f(Mx+v)+(a+c)\cdot x+b}\nonumber\\
&=\sum_y(-1)^{f(y)+(a+c)\cdot M^{-1}(y+v)+b}\nonumber\\
&=(-1)^{b+(a+c)\cdot M^{-1}v}W_f(M^{-T}(a+c)).
\end{align}
The frequency map $c\mapsto M^{-T}(a+c)$ is a bijection. Since $k+1$ is even, the sign factor disappears after taking the $(k+1)$-st power. Summing over all frequencies and applying~\eqref{firstformulaof pkf} proves $P_k(g)=P_k(f)$. For even $k$, the affine functions $0$ and $1$ have respective values $P_k(0)=0$ and $P_k(1)=1$, although they are EA-equivalent.
\end{proof}
 Since the restrictions of functions to affine spaces will be the main topic of this paper, let us recall from \cite{CSA2:2024} more detail on the ANF of such a restriction. The restriction of a function $f$ to an affine space $A$ will be
denoted by $f_{\mid A}$.
We can write $A=v+E$ for some vector space $E$ and some vector $v$. We then fix a basis for $E$ and use the coordinates in this basis to identify each element of $A$. This allows identifying $A$ with the set $\mathbb{F}_2^{n-k}$, where $k$ is the codimension of $A$, and the function $f_{\mid
A}$ can then be viewed as a function in $n-k$ variables, which allows to
consider its ANF. 
Different choices of bases for $E$
result in different expressions for the ANF of $f_{\mid A}$. While different choices of bases for $E$
result in different expressions for the ANF of $f_{\mid A}$, all of these functions are affinely equivalent and therefore have the same algebraic degree.

More precisely, 
we can consider the affine space $A$ as the set of solutions of $k$ affine equations, that are, after Gaussian elimination, $x_{i_1} = a_{i_1}(y), \ldots, x_{i_k} = a_{i_k}(y)$, where $i_1, \ldots, i_k$ are distinct and $a_j(y)$ are affine functions in  $y=(y_1, \ldots, y_{n-k})$, where the $n-k$ variables $\{y_1, \ldots, y_{n-k}\} = \{x_1, \ldots, x_n\} \setminus \{x_{i_1}, \ldots, x_{i_k}\}$. We obtain one of the expressions for the ANF of $f_{\mid A}$ by substituting the variables $x_{i_1}, \ldots, x_{i_k}$ with $a_{i_1}(y), \ldots, a_{i_k}(y)$ in the ANF of $f$, and we obtain a function in the remaining $n-k$ variables $\{y_1, \ldots, y_{n-k}\}$.  This function depends on the equations used for defining $A$ (the choice of $i_1, \ldots, i_k$ is not unique),
 but all choices yield functions which are affinely equivalent to each other.
 
 The behavior of restrictions with respect to EA-equivalence is as follows.
 \begin{proposition}\textup{\cite[Lemma~4]{CSA2:2024}}\label{prop:restr-equivalence}
Let $f,g$ be EA-equivalent Boolean functions in $n$ variables, with $g=f\circ\varphi+\ell$ for an affine automorphism $\varphi$ and an affine Boolean function $\ell$. For every affine subspace $A\subseteq\F_2^n$, $g_{\mid A}$ and $f_{\mid\varphi(A)}$ are EA-equivalent.
\end{proposition}
\begin{proof}
An affine parametrization $\psi$ of $A$ gives $g\circ\psi=f\circ(\varphi\circ\psi)+\ell\circ\psi$. The map $\varphi\circ\psi$ parametrizes $\varphi(A)$, and $\ell\circ\psi$ is affine. This is precisely EA-equivalence of the restricted functions.
\end{proof}
\begin{remark}
To evaluate $P_k(f_{\mid A})$, the affine space $A$ is first identified with
$\F_2^{\dim A}$ through an affine isomorphism.  The sums occurring in the
$k$-th-order test are therefore taken in this coordinate vector space, not in
the ambient affine coset $A$.  When $k$ is odd, different affine
parametrizations give EA-equivalent functions and hence the same value of
$P_k$ by Lemma~\ref{lem: EA equiv-nonhomomor}.  For even $k$, this
independence need not hold, so throughout the restriction results involving
EA-invariance we assume that $k$ is odd.
\end{remark}
\begin{remark}\label{rem:bent-restriction}
A Boolean function $f$ in an even number $n$ of variables is bent if every Walsh coefficient has magnitude $2^{n/2}$. On any affine hyperplane $H$, the restriction identity in Theorem~\textup{\ref{thm:restriction-walsh-identity}} expresses a restricted coefficient, up to sign, as half the sum or difference of two ambient coefficients. Thus $W_{f_{\mid H}}$ takes values in $\{0,\pm2^{n/2}\}$. Such a function in $n-1$ variables is called semi-bent, or nearly bent. Its normalized nonzero coefficients have magnitude $2^{1-n/2}$, and Parseval shows that exactly $2^{n-2}$ coefficients are nonzero.

For every odd $k\ge3$, equation~\eqref{firstformulaof pkf} now gives
\[
P_k(f)=\frac12\left(1-2^{-n(k-1)/2}\right),\qquad
P_k(f_{\mid H})=\frac12\left(1-2^{-(n-2)(k-1)/2}\right).
\]
The decrease is the same on every hyperplane and equals
$P_k(f)-P_k(f_{\mid H}) =\frac12(2^{k-1}-1)2^{-n(k-1)/2}$.
For fixed $k$, this difference tends to zero as the even dimension grows.
\end{remark}
\begin{definition}\label{def:restriction-nonhomomorphicity}
Let $f$ be a Boolean function in $n$ variables, let $k\ge3$ be odd, and let $0\le t<n$. Define
$P_{k,t}(f)=\min_{\codim A\le t}P_k(f_{\mid A})$,
where $A$ ranges over the nonempty affine subspaces of $\F_2^n$. The odd-order assumption makes this definition independent of the affine parametrization of $A$.
\end{definition}
The restriction parameter inherits EA-invariance from $P_k$.
\begin{proposition}\label{prop:rnnl-equivalence}
Let $f,g$ be EA-equivalent Boolean functions in $n$ variables. For odd $k\ge3$ and $0\le t<n$, $P_{k,t}(f)=P_{k,t}(g)$.
\end{proposition}
\begin{proof}
Write $g=f\circ\varphi+\ell$, where $\varphi$ is an affine automorphism and $\ell$ is affine. Proposition~\textup{\ref{prop:restr-equivalence}} and Lemma~\textup{\ref{lem: EA equiv-nonhomomor}} give $P_k(g_{\mid A})=P_k(f_{\mid\varphi(A)})$ for every affine subspace $A$. The map $A\mapsto\varphi(A)$ preserves dimension and is bijective, with inverse $B\mapsto\varphi^{-1}(B)$. Thus the two families of values $\{P_k(g_{\mid A})\mid\codim A\le t\}$ and $\{P_k(f_{\mid B})\mid\codim B\le t\}$ coincide. Their minima are $P_{k,t}(g)$ and $P_{k,t}(f)$, respectively, proving equality.
\end{proof} 
\begin{lemma}\label{lem:exact-codimension-minima}
For odd $k\ge3$ and $0\le t<n$, the minimum defining $P_{k,t}(f)$ is attained on an affine subspace of codimension exactly $t$.
\end{lemma}
\begin{proof}
It suffices to show that every positive-dimensional restriction has a hyperplane restriction with no larger $P_k$. Represent it by a Boolean function $F(y,z)$, where $z\in\F_2$ and $y\in\F_2^{d-1}$, and put $F_e(y)=F(y,e)$. If $a=\NW_F(u,0)$ and $b=\NW_F(u,1)$, then direct summation over $z$ gives $\NW_{F_e}(u)=a+(-1)^e b$. Put $q=k+1$, which is even. The binomial theorem gives
\[
\frac{(a+b)^q+(a-b)^q}{2}
=a^q+b^q+\sum_{j=1}^{q/2-1}\binom q{2j}a^{q-2j}b^{2j}
\ge a^q+b^q.
\]
Sum over $u$. The average of the two restricted Walsh moments is at least the moment of $F$. Hence at least one of $F_0,F_1$ has no larger $P_k$, by~\eqref{firstformulaof pkf}. A minimizing subspace exists because the family of affine subspaces is finite. Starting with one of codimension $c\le t$, repeat the preceding construction $t-c$ times. At each step the dimension falls by one, and the value of $P_k$ does not increase. The process is possible because the final dimension is $n-t\ge1$. It gives a subspace $B$ of codimension exactly $t$ with $P_k(f_{\mid B})\le P_{k,t}(f)$. Conversely, $B$ belongs to the family defining $P_{k,t}(f)$, so $P_{k,t}(f)\le P_k(f_{\mid B})$. Both inequalities are equalities, proving attainment at codimension $t$.
\end{proof}

Note that 
\[P_{k}(f) = P_{k,0}(f)\ge P_{k,1}(f)\ge P_{k,2}(f)\ge \ldots \ge P_{k,n-1}(f)=0.\]
 Affine functions have $P_{k,t}=0$. Theorem~\ref{thm:strict-nonhomomorphicity} below proves that equality $P_k(f)=P_{k,1}(f)$ occurs only for affine functions. The bent-function example shows, however, that the decrease can be small.

 We recall the values of the $k$-th-order nonhomomorphicity for some simple functions. Before, let us recall that 
when $f$ is a  monomial function of algebraic degree $r\ge 2$,\begin{equation}\label{eq:nnl-monomial} \NW_f(a)=\delta_0(a)-2^{1-r}(-1)^{\w(a)}\end{equation}
 we have 
\begin{equation}\label{eq:nnl-monomial1}
P_k(f)=\frac12\left(1-\left((1-2^{1-r})^{k+1}+(2^r-1)2^{(k+1)(1-r)}\right)\right).
\end{equation}
Indeed, the term corresponding to $a=0$ is $(1-2^{1-r})^{k+1}$, while each of the other $2^r-1$ terms equals $2^{(k+1)(1-r)}$, since $k+1$ is even.
 When a function $f$ is a direct sum of two functions $g$ and $h$ in $n_1$ and $n_2$ variables, respectively, the normalized Walsh transform factorizes.
$\NW_f(a,b)=\NW_g(a)\NW_h(b)$.
Therefore, for odd $k$,
\begin{equation}\label{eq:nnl-direct-sum}
1-2P_k(f)=\bigl(1-2P_k(g)\bigr)\bigl(1-2P_k(h)\bigr),
\end{equation}
or equivalently,
\begin{equation}\label{eq:nnl-direct-sum-p}
P_k(f)=P_k(g)+P_k(h)-2P_k(g)P_k(h).
\end{equation}
More generally, when $f=f_1+\cdots+f_m$ is a direct sum of $m$ functions, repeated use of~\eqref{eq:nnl-direct-sum} gives
\begin{equation}\label{eq:nnl-direct-sum-many}
1-2P_k(f)=\prod_{i=1}^m\bigl(1-2P_k(f_i)\bigr).
\end{equation}
Equivalently,
\[
P_k(f)=\frac12\left(1-\prod_{i=1}^m\left(\sum_{a_i\in\F_2^{n_i}}\NW_{f_i}(a_i)^{k+1}\right)\right).
\]
Note that from Relation~\eqref{firstformulaof pkf}, $-1\leq \sum_{a\in\F_2^n}\NW_{f}(a)^{k+1}\leq 1$  meaning for $k$ odd that $0\leq \sum_{a\in\F_2^n}\NW_{f}(a)^{k+1}\leq 1$ and in this case we have from Relation~\eqref{eq:nnl-direct-sum-many}, for any subset $I\subseteq \{1, \ldots, m\}$, 
\begin{equation}\label{eq:nnl-direct-sum-inequality}
    P_k(f) \ge P_k\left(\sum_{i\in I}f_i\right).
\end{equation}
For a Boolean function in specified coordinates, $\var(f)$ denotes the set of indices of variables on which $f$ depends. For an affine hyperplane equation, $\var(H)$ is the support of its nonzero normal vector. A direct sum means that the summands have disjoint variable blocks. Unless unused variables are explicitly mentioned, a monomial direct sum uses all the variables in its monomials, so $n=\sum_i r_i$.

\section[Odd-order nonhomomorphicity of direct sums and restrictions]{$k$-th-order nonhomomorphicity of direct sums of monomials and restrictions}\label{sec:direct-sums-restrictions}

Recall that when $f$ is a 
monomial function of algebraic degree $r\ge 2$,  
 we have $W_f(a)=2^r\delta_0(a)-2(-1)^{\w(a)}\in\{2^r-2,\pm2\}$ which means $\NW_f(a)=\delta_0(a)-2^{1-r}(-1)^{\w(a)}$, where $\delta_0(a)=1$ for $a=0$ and $0$ otherwise.
 When a function $f$ is a direct sum of two functions $g$ and $h$ in $n$ and $m$ variables, respectively (that is, $f=g+h$ and $\var(g)\cap\var(h) = \emptyset$), it is known that $W_f(a,b) = W_g(a)W_h(b)$ for all $a\in\F_2^n$ and $b\in\F_2^m$ (see~\cite{cusickStanica2017} for example). This means we also have $\NW_f(a,b) = \NW_g(a)\NW_h(b)$. Hence, for $k$ odd, from Relation~\eqref{eq:nnl-monomial1}   the relation  \begin{eqnarray}\label{eq:fst-expre-sum-nhom}
     \sum_{a_i\in\F_2^{r_i}} \left(\delta_0(a_i)-2^{1-r_i}(-1)^{\w(a_i)}\right)^{k+1}=\left((1-2^{1-r_i})^{k+1}+(2^{r_i}-1)2^{(k+1)(1-r_i)}\right)\end{eqnarray}
 holds and  the following identity holds. 
 \begin{lemma}Let $k\ge3$ be odd and let $f=f_1+\cdots+f_m$ be  a direct sum of $m$ monomials of degree at least two with $\deg(f_i)=r_i$. We have \begin{align}\label{eq:DSM-pkf}
     P_{k,0}(f)=P_k(f) &=\dfrac{1}{2}\left(1-\sum_{(a_1,\ldots,a_m)\in\prod_i\F_2^{r_i}}\left( \prod_{i=1}^m\left(\delta_0(a_i)-2^{1-r_i}(-1)^{\w(a_i)}\right)\right)^{k+1}\right)\nonumber\\&= \frac{1}{2}\left(1-\prod_{i=1}^m\sum_{a_i\in\F_2^{r_i}} \left(\delta_0(a_i)-2^{1-r_i}(-1)^{\w(a_i)}\right)^{k+1}\right)\\&=\dfrac{1}{2}\left(1- \prod_{i=1}^{m}\left((1-2^{1-r_i})^{k+1}+(2^{r_i}-1)2^{(k+1)(1-r_i)}\right)\right)\nonumber.
\end{align}
 \end{lemma} 
Note that from~\eqref{eq:fst-expre-sum-nhom},  if we set $S(r) = \left(1 - 2^{1-r}\right)^k + (2^r - 1) 2^{k(1-r)}$ the following identity holds.
\begin{lemma}\label{prop:encrease-sum}
  Let $k \geq 2$ and $r \geq 2$ be integers. Define
$S(r) = \left(1 - 2^{1-r}\right)^k + (2^r - 1) 2^{k(1-r)}$.
Then the function $S$ is increasing with respect to $r$, that is,
\[
 S(r+1) \geq S(r) \quad \text{for all } r \geq 2.
\]
\end{lemma}
\begin{proof}
 Set $x=2^{r-1}$, so that $x\ge2$.
Then $2^{1-r} = \frac{1}{x}$ and $2^r = 2x$, hence
\[
 S(r) = \left(1 - \frac{1}{x}\right)^k + \frac{2x - 1}{x^k}
 = \frac{(x-1)^k + 2x - 1}{x^k}.
\]
Define
$g(x) = \frac{(x-1)^k + 2x - 1}{x^k}$.
Then $S(r) = g(2^{r-1})=g(x)$ and $S(r+1) = g(2^{r})=g(2x)$. Since $2^{r}$ is increasing in $r$, it suffices to prove that
$g(2x)\ge g(x)$ for all $x\ge2$. 
We compute
\[
 g(2x) - g(x)
 = \frac{(2x-1)^k + 4x - 1}{2^k x^k} - \frac{(x-1)^k + 2x - 1}{x^k}.
\]
Multiplying both sides by $2^k x^k > 0$, we define
\[
 \Delta = (2x-1)^k - 2^k (x-1)^k - (2^{k+1} - 4)x + (2^k - 1).
\]
Let $y = x - 1 \geq 1$. Then $2x - 1 = 2y + 1$, and
\begin{equation}\label{eq:delta-sum-dsm}
 \Delta = (2y+1)^k - 2^k y^k - (2^{k+1} - 4)(y+1) + (2^k - 1).
\end{equation}
Expanding $(2y+1)^k$ using the binomial theorem,
$(2y+1)^k = \sum_{j=0}^k \binom{k}{j} (2y)^j$.
The term corresponding to $j = k$ cancels with $-2^k y^k$ in~\eqref{eq:delta-sum-dsm}, which  becomes
\[
 \Delta = \sum_{j=0}^{k-1} \binom{k}{j} (2y)^j - (2^{k+1} - 4)y + (3 - 2^k).
\] 
Since $y \geq 1$, we have $y^j \geq y$ for all $j \geq 1$. Therefore,
\begin{equation}\label{ineq:sum-dsm-F2n}
 \Delta \geq 1  + y \sum_{j=1}^{k-1} \binom{k}{j} 2^j - (2^{k+1} - 4)y + (3 - 2^k).
\end{equation}
Using the identity
$\sum_{j=0}^k \binom{k}{j} 2^j = 3^k$,
we obtain
$\sum_{j=1}^{k-1} \binom{k}{j} 2^j = 3^k - 1 - 2^k$.
Substituting this in~\eqref{ineq:sum-dsm-F2n} yields
\[
 \Delta \geq y (3^k - 3 \cdot 2^k + 3) + (4 - 2^k).
\]
The coefficient $3^k-3\cdot2^k+3$ of $y$ is zero when $k=2$. For $k\ge3$, $(3/2)^k\ge27/8>3$, so that coefficient is positive. Since $y\ge1$, replacing $y$ by one therefore gives
\[
 \Delta \geq 3^k - 4 \cdot 2^k + 7.
\]
Define
$h(k) = 3^k - 4 \cdot 2^k + 7$.
One can easily check that $h(2) = 0$ and $h(3)=2 > 0$. Moreover,
\[
 h(k+1) - h(k) = 2 \cdot 3^k - 4 \cdot 2^k = 2(3^k - 2^{k+1}) > 0 \quad \text{for } k \geq 3.
\]
Thus $h(k) \geq 0$ for all $k \geq 2$, and therefore $\Delta \geq 0$.

Hence $g(2x)\ge g(x)$ for all $x\ge2$, so $S(r+1)\ge S(r)$.
This completes the proof.
\end{proof}
For later use, denote the monomial moment factor by
$A_r(q)=(1-2^{1-r})^q+(2^r-1)2^{q(1-r)}$.
Thus $A_r(q)=M_q(x_1\cdots x_r)$ for even $q\ge2$.

\subsection[Lower bounds of the restriction parameter]{Lower bounds of $P_{k,t}$}
 The direct-sum formula gives the following lower bound for $P_{k,0}(f)$.
 \begin{lemma} Let $f=f_1+\cdots +f_m$ be a function in $n$ variables that is the direct sum of the monomials $f_1, \ldots, f_m$ with $m\ge 2$ and $\deg(f_i)=r_i\geq 2$ for all $i\in\{1,\ldots,m\}$. For $k\geq 2$ odd,  we have $$P_{k,0}(f)\geq\dfrac{1}{2}\left(1-2^n\left( \prod_{i=1}^m\left(1-2^{1-r_i}\right)\right)^{k+1}\right).$$   
 \end{lemma}
\begin{proof}
For a monomial of degree $r\ge2$, the zero-frequency normalized coefficient is $1-2^{1-r}$ and every nonzero coefficient has magnitude $2^{1-r}\le1-2^{1-r}$. Multiplicativity therefore gives $|\NW_f(u)|\le\NW_f(0)=\prod_i(1-2^{1-r_i})$ at every frequency. There are $2^n$ frequencies. Bound every term in~\eqref{firstformulaof pkf} by this maximum raised to $k+1$ to obtain the result.
\end{proof}
 Now we consider the restriction $f_{|A}$ of $f$ to a affine subspace $A$ of codimension $t$.
 We use the technique as in \cite{SergeAnaPante2025} to provide a lower bound of $P_{k,t}(f)$. We start with $t=1$.
\begin{proposition}\label{pro:non-hom-dsm-hyp} Let $n$ and $k$ be two integers with $k\geq 2$ odd.  Let $f=f_1+\cdots +f_m$ be a function in $n$ variables that is the direct sum of the monomials $f_1, \ldots, f_m$ with $m\ge 2$, $\deg(f_1)\ge\ldots \ge \deg(f_m)\ge 2$ and $\deg(f_i)=r_i\geq 2$ for all $i\in\{1,\ldots,m\}$. 
Then, for every hyperplane $H$,
$$P_{k}(f_{|H})\geq\dfrac{1}{2}\left(1-2^{n-1}\left( \prod_{i=1}^{m-1}\left(1-2^{1-r_i}\right)\right)^{k+1}\right),$$
In particular,
$$P_{k,1}(f)\geq\dfrac{1}{2}\left(1-2^{n-1}\left( \prod_{i=1}^{m-1}\left(1-2^{1-r_i}\right)\right)^{k+1}\right).$$
 \end{proposition}
 \begin{proof}
Write $H=\{x\mid u\cdot x=\varepsilon\}$ with $u\ne0$. If the equation meets a monomial block, choose such a block $f_j$ and solve for one of its variables. Put $g=f_j$ and $h=\sum_{i\ne j}f_i$. Lemma~\textup{\ref{lem:ineq-nwrest}} shows that each restricted coefficient has magnitude at most $\max_v|\NW_h(v)|$. For monomial direct sums, this maximum is
\[
\NW_h(0)=\prod_{i\ne j}(1-2^{1-r_i})
\le\prod_{i=1}^{m-1}(1-2^{1-r_i}).
\]
The last inequality uses $r_1\ge\cdots\ge r_m$, so the product of $m-1$ factors is largest when the smallest factor is omitted. If the equation meets only unused variables, the restriction represents $f$ with fewer unused variables, and its largest normalized coefficient is instead $\prod_{i=1}^m(1-2^{1-r_i})$, which satisfies the same upper bound.

There are $2^{n-1}$ restricted coefficients. Because $k+1$ is even, equation~\eqref{firstformulaof pkf} gives
\[
\begin{aligned}
P_k(f_{\mid H})
&=\frac12\left(1-\sum_{u'}|\NW_{f_{\mid H}}(u')|^{k+1}\right)\\
&\ge\frac12\left(1-2^{n-1}\left(\prod_{i=1}^{m-1}(1-2^{1-r_i})\right)^{k+1}\right).
\end{aligned}
\]
Taking the minimum over hyperplanes and using Lemma~\textup{\ref{lem:exact-codimension-minima}} proves the bound for $P_{k,1}(f)$.
\end{proof}
In general, we have
 \begin{proposition}\label{pro:non-hom-dsm-aff} Let $n$, $t$ and $k$ be three  integers with $0\le t<m$ and $k\ge3$ odd.  Let $f=f_1+\cdots +f_m$ be a function in $n$ variables that is the direct sum of the monomials $f_1, \ldots, f_m$ with $m\ge 2$, $\deg(f_1)\ge\ldots \ge \deg(f_m)\ge 2$ and $\deg(f_i)=r_i\geq 2$ for all $i\in\{1,\ldots,m\}$. Then, for every affine space  $A$ of codimension $t$, $$P_k(f_{|A})\geq\dfrac{1}{2}\left(1-2^{n-t}\left( \prod_{i=1}^{m-t}\left(1-2^{1-r_i}\right)\right)^{k+1}\right),$$
 In particular, 
 $$P_{k,t}(f)\geq\dfrac{1}{2}\left(1-2^{n-t}\left( \prod_{i=1}^{m-t}\left(1-2^{1-r_i}\right)\right)^{k+1}\right),$$
 \end{proposition} 
 \begin{proof}
Let $A$ have codimension exactly $t$. Apply Lemma~\textup{\ref{lem:ineq-nwrest-aff}} to choose the pivots. At most $t$ monomial blocks are affected, so at least $m-t$ remain in the untouched sum $h$. The coefficient estimate in that lemma gives
\begin{equation}\label{eq1:general-form-rest-k}
\max_{u,v}|\NW_{f_{\mid A}}(u,v)|\le\max_v|\NW_h(v)|.
\end{equation}
For a direct sum of monomials the right side is $\NW_h(0)$. Each zero-frequency factor belongs to $(0,1)$, so deleting factors increases their product. Among products of $m-t$ factors, the largest uses the $m-t$ highest degrees. Therefore
\begin{equation}\label{eq2:general-form-rest-k}
\NW_h(0)\le\prod_{i=1}^{m-t}(1-2^{1-r_i}).
\end{equation}
This argument also covers a pivot in an unused variable, which need not affect any monomial. The $2^{n-t}$ restricted coefficients and~\eqref{firstformulaof pkf} now give
\[
P_k(f_{\mid A})\ge\frac12\left(1-2^{n-t}\left(\prod_{i=1}^{m-t}(1-2^{1-r_i})\right)^{k+1}\right).
\]
By Lemma~\textup{\ref{lem:exact-codimension-minima}}, the minimum over codimensions at most $t$ equals the minimum at codimension exactly $t$. This proves the asserted bound for $P_{k,t}(f)$.
\end{proof}


To refine the hyperplane bound, we first estimate restricted Walsh coefficients for a direct sum of two functions.
\begin{lemma}\textup{\cite[Proof of Theorem~3]{SergeAnaPante2025}}\label{lem:ineq-nwrest}
Let $f(x,z)=g(x)+h(z)$, where $x=(x_1,y)\in\F_2^r$ and $z\in\F_2^{n-r}$. Suppose a hyperplane $H$ is defined by
\begin{equation}\label{eq:generalresultforDS-1-lem}
x_1=a\cdot y+b\cdot z+\varepsilon.
\end{equation}
Write $g(x_1,y)=x_1g_1(y)+g_2(y)$. For each $u\in\F_2^{r-1}$ and $v\in\F_2^{n-r}$,
\[
|\NW_{f_{\mid H}}(u,v)|
\le 2^{-(r-1)}\sum_y|\NW_h(v+g_1(y)b)|.
\]
\end{lemma}
\begin{proof}
After substituting~\eqref{eq:generalresultforDS-1-lem}, the restriction is
\[
F(y,z)=g_2(y)+(a\cdot y+\varepsilon)g_1(y)+g_1(y)b\cdot z+h(z).
\]
Set $G_0(y)=g_2(y)+(a\cdot y+\varepsilon)g_1(y)$. Summing first over $z$ gives
\[
\NW_F(u,v)=2^{-(r-1)}\sum_y(-1)^{G_0(y)+u\cdot y}\NW_h(v+g_1(y)b).
\]
The triangle inequality proves the stated coefficient bound.
\end{proof}
\begin{theorem}
\label{thm:genralresultforDSf} 
Let $f$ be the direct sum of two functions, $f(x_1,\ldots,x_n) = g(x_1,\ldots,x_r)+ h(x_{r+1}, \ldots, x_n).$
Let $H$ be a hyperplane such that the equation which defines $H$ contains at least one of the variables that $g$ depends on (that is, $\var(g)\cap\var(H)\neq \emptyset$). Then for $k$ odd, $$\displaystyle P_{k}(f_{|H}) \ge \dfrac{1}{2}\left(1- 2^{r-1}\sum_{\mathbf{u}^{(2)}\in \F_2^{n-r}}\NW_{h}(\mathbf{u}^{(2)})^{k+1}\right).$$
It follows that $P_{k,1}(f)\ge \min\left(\dfrac{1}{2}\left(1- 2^{n-r-1}\sum_{\mathbf{u}^{(1)}\in \F_2^{r}}\NW_{g}(\mathbf{u}^{(1)})^{k+1}\right),\dfrac{1}{2}\left(1- 2^{r-1}\sum_{\mathbf{u}^{(2)}\in \F_2^{n-r}}\NW_{h}(\mathbf{u}^{(2)})^{k+1}\right)\right)$.
\end{theorem}
\begin{proof}
Suppose the hyperplane equation meets the $g$ block. Solve it for one variable of that block, and use the notation of Lemma~\textup{\ref{lem:ineq-nwrest}}. Put $q=k+1$. Since $q$ is even and $x\mapsto x^q$ is convex on $[0,\infty)$, the coefficient estimate and Jensen's inequality give
\[
|\NW_{f_{\mid H}}(u,v)|^q
\le2^{-(r-1)}\sum_y|\NW_h(v+g_1(y)b)|^q.
\]
Sum over the $2^{r-1}$ choices of $u$, and then over $v$, to obtain
\[
\begin{aligned}
M_q(f_{\mid H})
&\le\sum_y\sum_v|\NW_h(v+g_1(y)b)|^q\\
&=2^{r-1}\sum_v|\NW_h(v)|^q.
\end{aligned}
\]
For the equality, for every fixed $y$, translation by $g_1(y)b$ permutes the $v$ frequencies. Equation~\eqref{firstformulaof pkf} proves the conditional bound.

Every hyperplane equation meets at least one of the two variable blocks. If it meets the $g$ block, apply the bound just proved. If it meets only the $h$ block, interchange $g$ and $h$, replacing $r$ by $n-r$. Taking the smaller of the two resulting lower bounds gives a bound valid for every hyperplane. Lemma~\textup{\ref{lem:exact-codimension-minima}} then gives the stated lower bound for $P_{k,1}(f)$.
\end{proof}

\begin{remark}\label{rem:no-universal-moment-ordering}
For convenience, put
$A_r(k)=\left(1-2^{1-r}\right)^{k+1} +(2^r-1)2^{(k+1)(1-r)}$.
There is no degree ordering that minimizes the expression below for every odd
$k$.  In particular, the sequence $2^{-r}A_r(k)$ need not be monotone.  For
$k=5$, for example,
$2^{-2}A_2(5)=1/64$, whereas $2^{-3}A_3(5)=23/1024>1/64$.
Thus the subset minimum in the next corollary should not, in general, be
replaced by the choice obtained solely by ordering the degrees.  This point is
important only for the explicit lower bound. The exact result proved later by
the contraction argument has a different and valid degree ordering.
\end{remark}

\begin{corollary}\label{coro: better-pk1-lower-bound}
Let $k\ge3$ be odd, and let $f=f_1+\cdots+f_m$ be a direct sum of
monomials, where $\deg(f_i)=r_i\ge2$. Then
\[
P_{k,1}(f)\ge
\min_{1\le j\le m}
\frac12\left(
1-2^{r_j-1}\prod_{i\ne j}
\left((1-2^{1-r_i})^{k+1}
 +(2^{r_i}-1)2^{(k+1)(1-r_i)}\right)
\right).
\]
\end{corollary}
\begin{proof}
First suppose the hyperplane equation contains a variable in a monomial block $f_j$. Choose such a $j$, put $g=f_j$, and put $h=\sum_{i\ne j}f_i$. The conditional part of Theorem~\textup{\ref{thm:genralresultforDSf}} gives
$P_k(f_{\mid H})\ge\frac12\left(1-2^{r_j-1}M_{k+1}(h)\right)$.
The direct-sum factorization gives
\[
M_{k+1}(h)=\prod_{i\ne j}\left((1-2^{1-r_i})^{k+1}+(2^{r_i}-1)2^{(k+1)(1-r_i)}\right).
\]
The index $j$ depends on $H$. Thus a lower bound valid for every such $H$ is the minimum over all $j$, not the bound for a fixed $j$.

If the equation involves only unused variables, the restricted moment equals $M_{k+1}(f)$. Since every block moment is at most $1$, this is at most $M_{k+1}(h)$ and hence at most $2^{r_j-1}M_{k+1}(h)$ for any $j$. The same minimum lower bound still applies. Finally take the minimum over hyperplanes, using Lemma~\textup{\ref{lem:exact-codimension-minima}}.
\end{proof}
\begin{remark}
The bound in Corollary~\textup{\ref{coro: better-pk1-lower-bound}} is at least as
strong as the crude bound in Proposition~\textup{\ref{pro:non-hom-dsm-hyp}}.  Indeed,
for every choice of $j$, if $h=\sum_{i\ne j}f_i$, then
$\sum_u |\NW_h(u)|^{k+1} \le 2^{n-r_j}|\NW_h(0)|^{k+1}$,
and $|\NW_h(0)|$ is no larger than the product of the $m-1$ largest
zero-frequency factors $1-2^{1-r_i}$.
\end{remark}

For higher codimension, we need a pivot choice that preserves affine dependence on the untouched blocks.
\begin{lemma}\textup{\cite[Lemma~2 and Proof of Theorem~3]{SergeAnaPante2025}}\label{lem:ineq-nwrest-aff}
Let $f=f_1+\cdots+f_m$ be a direct sum on disjoint variable blocks, and let $A$ be an affine subspace of codimension $t$. There is a choice of $t$ distinct pivot variables $x_{i_1},\ldots,x_{i_t}$ and affine equations $x_{i_j}=a_j(y,z)$ defining $A$ with the following property. Let $g$ be the sum of the blocks containing pivots and let $h$ be the sum of the untouched blocks. Write $y$ for the nonpivot variables of $g$ and $z$ for the variables of $h$, of dimensions $n_1$ and $n_2$. Then $n_1+n_2=n-t$, at most $t$ blocks occur in $g$, and
\begin{equation}\label{general-form-rest-k-lem}
f_{\mid A}(y,z)=G_0(y)+G(y)\cdot z+h(z).
\end{equation}
Here $G_0\colon\F_2^{n_1}\to\F_2$ and $G\colon\F_2^{n_1}\to\F_2^{n_2}$. Moreover,
\[
|\NW_{f_{\mid A}}(u,v)|\le2^{-n_1}\sum_y|\NW_h(v+G(y))|.
\]
\end{lemma}
\begin{proof}
Put $c=t$.
Represent $A$ by $c$ independent affine equations. At each elimination step, inspect the entire residual coefficient matrix, consisting of all rows not yet used as pivot rows. If it has a nonzero column belonging to a block that has not yet supplied a pivot, choose such a column and then a row in which it is nonzero. Otherwise choose any nonzero residual column. Interchange rows as needed, isolate the pivot variable, and eliminate it from the remaining rows. Independence of the equations ensures that this process produces $c$ distinct pivots.

All first pivots from different blocks are chosen before any repeated pivot from a block. To see this, at the first repeated choice every column in an as yet untouched block is zero throughout the residual matrix. Subsequent row operations among residual rows preserve these zero columns. Thus no later step can introduce a pivot from a new block.

Let the pivots, in elimination order, be $v_1,\ldots,v_c$, and let $r$ be the number of affected blocks. The first $r$ pivots belong to distinct blocks, and every $v_j$ with $j>r$ is a repeated pivot. Let $z$ denote the variables of the blocks that remain untouched. If $r<c$, all $z$ columns vanish in every residual row when $v_{r+1}$ is chosen. A later triangular equation therefore has the form $v_j=\alpha_j(y)+\sum_{\ell>j}d_{j\ell}v_\ell$ for $j>r$, with $\alpha_j$ affine and $d_{j\ell}\in\F_2$. Starting at $j=c$ and descending to $j=r+1$, substitution proves that each such $v_j$ is affine in $y$ alone. Substituting these expressions and then eliminating the later first pivots expresses each $v_j$ with $j\le r$ as $\alpha_j(y)+b_j\cdot z$, where $\alpha_j$ is affine and $b_j$ is a constant vector. If $r=c$, there are no repeated pivots and this last form follows directly by triangular back substitution.

Fix an affected block. At most one of its substituted pivot variables, namely its first pivot, has an expression involving $z$. Every Boolean function is affine in each individual input when the other inputs are fixed. More explicitly, if this distinguished input is $v$, the block can be written as $vB+C$, where $B$ and $C$ do not involve $v$. After substituting all other pivots, $B$ and $C$ depend only on $y$. For the block whose first pivot is $v_j$, write the resulting functions as $B_j(y)$ and $C_j(y)$. Substitution gives $(\alpha_j(y)+b_j\cdot z)B_j(y)+C_j(y)$. Thus one can take $G_0(y)=\sum_{j=1}^r(\alpha_j(y)B_j(y)+C_j(y))$ and $G(y)=\sum_{j=1}^rB_j(y)b_j$. These formulas give~\eqref{general-form-rest-k-lem} after the untouched sum $h(z)$ is added. There are $c$ pivots, so at most $c$ blocks are affected. When $c=0$, take $g=0$ and $h=f$.

For fixed $u\in\F_2^{n_1}$ and $v\in\F_2^{n_2}$, sum first over the untouched variables $z$. This gives
\[
\NW_{f_{\mid A}}(u,v)
=2^{-n_1}\sum_y(-1)^{G_0(y)+u\cdot y}\NW_h(v+G(y)).
\]
Taking absolute values and applying the triangle inequality proves the coefficient estimate.
\end{proof}
\begin{theorem}\label{thm:genralresultforDSfaff}
Let $f(x,z)=g(x)+h(z)$, where $x\in\F_2^p$ and $z\in\F_2^{n-p}$. Let $A$ have codimension $t\le p$. Assume that $A$ has an affine parametrization by $(y,z)\in\F_2^{p-t}\times\F_2^{n-p}$ for which
\begin{equation}\label{eq:affine-coupling-hypothesis}
f_{\mid A}(y,z)=G_0(y)+G(y)\cdot z+h(z).
\end{equation}
For every odd $k\ge3$,
\[
P_k(f_{\mid A})\ge\frac12\left(1-2^{p-t}\sum_v|\NW_h(v)|^{k+1}\right).
\]
For an arbitrary direct sum $f=f_1+\cdots+f_m$ on disjoint blocks of positive dimensions $n_1,\ldots,n_m$, put $n=\sum_i n_i$. For every integer $0\le t<n$, put $s=\min(t,m)$. Then
\begin{equation}\label{eq:general-affine-dimension-bound}
P_{k,t}(f)\ge
\min_{\substack{J\subseteq\{1,\ldots,m\}\\|J|=s}}
\frac12\left(1-2^{p_J-t}\prod_{i\notin J}M_{k+1}(f_i)\right),
\qquad p_J=\sum_{j\in J}n_j.
\end{equation}
Empty products are $1$. In the second assertion, the pivot construction selects the affected blocks as in Lemma~\textup{\ref{lem:ineq-nwrest-aff}}.
\end{theorem}
\begin{proof}
Put $q=k+1$ and $a=p-t$. For $u\in\F_2^a$ and $v\in\F_2^{n-p}$, direct summation in the coupling form~\eqref{eq:affine-coupling-hypothesis} gives
\[
\NW_{f_{\mid A}}(u,v)
=2^{-a}\sum_y(-1)^{G_0(y)+u\cdot y}\NW_h(v+G(y)).
\]
The weights $2^{-a}$ in this sum are nonnegative and sum to one. The triangle inequality and convexity of $x\mapsto x^q$ on $[0,\infty)$ therefore give
\[
|\NW_{f_{\mid A}}(u,v)|^q
\le\left(2^{-a}\sum_y|\NW_h(v+G(y))|\right)^q
\le2^{-a}\sum_y|\NW_h(v+G(y))|^q.
\]
Summing over $u$ contributes $2^a$ identical upper bounds. For each fixed $y$, the substitution $w=v+G(y)$ is a bijection of $\F_2^{n-p}$. Consequently,
\[
\begin{aligned}
M_q(f_{\mid A})
&\le\sum_u2^{-a}\sum_y\sum_v|\NW_h(v+G(y))|^q\\
&=\sum_y\sum_w|\NW_h(w)|^q
=2^aM_q(h).
\end{aligned}
\]
Since $P_k(f_{\mid A})=\frac12(1-M_q(f_{\mid A}))$, subtracting this upper bound from one and dividing by two proves the first assertion. The calculation also applies when $a=0$, since then the sum over $y$ has one term.

We next prove the direct-sum assertion for $1\le t<m$. Fix any affine subspace $A$ of codimension exactly $t$. By Lemma~\textup{\ref{lem:ineq-nwrest-aff}}, choose its $t$ pivots so that, if $J_0$ denotes the set of blocks containing them, the restriction has a coupling form in the remaining variables of those blocks and the variables of all other blocks. Write $p_0=\sum_{j\in J_0}n_j$. All $t$ pivots lie among these $p_0$ variables, so $p_0\ge t$. Solving the independent pivot equations leaves exactly $p_0-t$ free variables in these blocks. Denote their tuple by $y_0$. Since each affected block contains a pivot, $|J_0|\le t$.

Choose $t-|J_0|$ further blocks and let $E$ be their index set. Such a choice exists because $m-|J_0|\ge t-|J_0|$. Put $J=J_0\cup E$, $K=\{1,\ldots,m\}\setminus J$, and $p_E=\sum_{j\in E}n_j$. Let $w\in\F_2^{p_E}$ denote the variables of the added blocks and $z\in\F_2^{n-p_J}$ those of the blocks in $K$. These were all untouched by the pivot construction. Splitting its coefficient vector according to $(w,z)$, the original coupling form is
\[
f_{\mid A}(y_0,w,z)
=G_0(y_0)+G_E(y_0)\cdot w+G_K(y_0)\cdot z+h_E(w)+h_K(z),
\]
where $h_E=\sum_{j\in E}f_j$ and $h_K=\sum_{i\in K}f_i$. Empty tuples and empty sums are allowed here. Define $Y=(y_0,w)$, set $\widetilde G_0(Y)=G_0(y_0)+G_E(y_0)\cdot w+h_E(w)$, and set $\widetilde G(Y)=G_K(y_0)$. Then
\[
f_{\mid A}(Y,z)=\widetilde G_0(Y)+\widetilde G(Y)\cdot z+h_K(z).
\]
Thus enlarging the affected set preserves the required coupling form. Its free-variable dimension is $(p_0-t)+p_E=p_J-t$, and its untouched-variable dimension is $n-p_J$. The first assertion, with $p=p_J$, gives $M_q(f_{\mid A})\le2^{p_J-t}M_q(h_K)$.

To compute the remaining moment, write a frequency on the $K$ blocks as $(v_i)_{i\in K}$. Disjointness of the blocks gives $\NW_{h_K}((v_i)_{i\in K})=\prod_{i\in K}\NW_{f_i}(v_i)$. Taking absolute values to the power $q$ and summing independently over these frequencies yields $M_q(h_K)=\prod_{i\in K}M_q(f_i)$. If $K$ is empty, $h_K$ is the zero function on the zero-dimensional space, whose sole normalized coefficient is one, so the same identity holds with an empty product.

Let $\mathcal J_t$ be the finite family of all $t$-element subsets of $\{1,\ldots,m\}$, and define $C(J)=2^{p_J-t}\prod_{i\notin J}M_q(f_i)$ for $J\in\mathcal J_t$. The construction above gives, for each $A$, one set $J_A\in\mathcal J_t$ such that $M_q(f_{\mid A})\le C(J_A)$. This set may depend on $A$, but $C(J_A)\le\max_{J\in\mathcal J_t}C(J)$. It follows that
\[
\begin{aligned}
P_k(f_{\mid A})
&\ge\frac12\left(1-\max_{J\in\mathcal J_t}C(J)\right)\\
&=\min_{J\in\mathcal J_t}\frac12\left(1-C(J)\right).
\end{aligned}
\]
The right side is independent of $A$. Taking the minimum over subspaces of codimension exactly $t$ and then using Lemma~\textup{\ref{lem:exact-codimension-minima}} proves~\eqref{eq:general-affine-dimension-bound} for $1\le t<m$.

If $t=0$, the only admissible subspace is the full space and the only index set is $J=\varnothing$. The right side is $\frac12(1-\prod_iM_q(f_i))=P_k(f)=P_{k,0}(f)$, by the same factorization. If $m\le t<n$, then $s=m$ and the only index set is $J=\{1,\ldots,m\}$. Its moment product is one and $p_J=n$, so the right side of~\eqref{eq:general-affine-dimension-bound} is $\frac12(1-2^{n-t})\le0$. Since every $P_k(f_{\mid A})$ is nonnegative, $P_{k,t}(f)\ge0$ proves the asserted inequality in this case. The universal bound zero is stronger than this last right side.
\end{proof}

\begin{remark}\label{rem:affine-coupling-necessity}
The coupling hypothesis cannot be replaced by the condition that every defining equation contains a variable of $g$. For example, let $g(x_1,x_2)=x_1x_2$ and $h(z_1,z_2)=z_1z_2$. On $A=\{x_1=z_1,\ x_2=z_2\}$ the two terms cancel, so $P_3(f_{\mid A})=0$. Here $M_4(h)=1/4$, which would give the incorrect lower bound $3/8$ if the coupling hypothesis were omitted. The pivot grouping in~\eqref{eq:general-affine-dimension-bound} allows both blocks to be affected. Its right side is then $\frac12(1-2^{4-2})=-3/2$, which is valid but weaker than the universal bound zero. That universal bound is attained here.
\end{remark}
\begin{corollary}\label{coro: better-pk1-lower-bound-aff}
Let $k\ge3$ be odd, and let $f=f_1+\cdots+f_m$ be a direct sum of
monomials, where $\deg(f_i)=r_i\ge2$.  For $1\le t<m$, we have
\[
P_{k,t}(f)\ge
\min_{\substack{I\subseteq\{1,\ldots,m\}\\ |I|=m-t}}
\frac12\left(
1-2^{\left(\sum_{j\notin I}r_j\right)-t}
\prod_{i\in I}
\left((1-2^{1-r_i})^{k+1}
 +(2^{r_i}-1)2^{(k+1)(1-r_i)}\right)
\right).
\]
\end{corollary}
\begin{proof}
For each affine subspace $A$ of codimension exactly $t$, choose its affected blocks as in Lemma~\textup{\ref{lem:ineq-nwrest-aff}} and extend them to a set $J$ of exactly $t$ blocks. Put $I=\{1,\ldots,m\}\setminus J$. Then $|I|=m-t$, the affected-variable dimension is $p_J=\sum_{j\notin I}r_j$, and the untouched sum is $h=\sum_{i\in I}f_i$. Theorem~\textup{\ref{thm:genralresultforDSfaff}} gives
\[
P_k(f_{\mid A})\ge\frac12\left(1-2^{p_J-t}\prod_{i\in I}M_{k+1}(f_i)\right).
\]
Put $q=k+1$, and for each $(m-t)$-element set $I$ let $L(I)=\frac12(1-2^{p_{I^c}-t}\prod_{i\in I}A_{r_i}(q))$, where $I^c=\{1,\ldots,m\}\setminus I$ and $p_{I^c}=\sum_{j\notin I}r_j$. Equation~\eqref{eq:fst-expre-sum-nhom} gives $M_q(f_i)=A_{r_i}(q)$, so the inequality just obtained is $P_k(f_{\mid A})\ge L(I_A)$ for the unaffected set $I_A$ chosen for that subspace. Since $L(I_A)\ge\min_{|I|=m-t}L(I)$, every codimension-$t$ subspace satisfies this same lower bound. Taking the minimum over those subspaces, and applying Lemma~\textup{\ref{lem:exact-codimension-minima}}, gives $P_{k,t}(f)\ge\min_{|I|=m-t}L(I)$. Expanding $A_{r_i}(q)$ is exactly the claimed formula. Here the blocks are the monomial variable blocks, with $n=\sum_i r_i$.
\end{proof}

\begin{remark}
As in Remark~\ref{rem:no-universal-moment-ordering}, the minimizing subset
in this lower bound cannot in general be identified solely from the ordering
$r_1\ge\cdots\ge r_m$.  The finite minimum displayed above is the correct
general form.  It is at least as strong as Proposition~\ref{pro:non-hom-dsm-aff},
since for every admissible $I$,
\[
\sum_u|\NW_{\sum_{i\in I}f_i}(u)|^{k+1}
\le 2^{\sum_{i\in I}r_i}
 |\NW_{\sum_{i\in I}f_i}(0)|^{k+1}.
\]
\end{remark}

\subsection[Exact value of the restriction parameter]{Exact value of $P_{k,t}$}
We now determine the restriction minimum exactly.  The proof uses a contraction property of the normalized Fourier transform.  We use the Walsh moments $M_q$ defined in the preliminaries.
When $q=k+1$ and $k$ is odd, Relation~\eqref{firstformulaof pkf} gives
$P_k(F)=\frac12\left(1-M_q(F)\right)$.

\begin{lemma}
\label{lem:fourier-contraction}
Let $V=\F_2^d$, let $\Phi\colon V\to\mathbb R$, and define
$\widehat{\Phi}(u)=2^{-d}\sum_{x\in V}\Phi(x)(-1)^{u\cdot x}$.
For every real $q\ge2$,
\[
  \sum_{u\in V}|\widehat{\Phi}(u)|^q
  \le
  2^{-d}\sum_{x\in V}|\Phi(x)|^q .
\]
\end{lemma}

\begin{proof}
For a finite vector $v$ and $p\ge1$, write $\|v\|_p=(\sum_i|v_i|^p)^{1/p}$. Let $N=2^d$, and let $U$ be the orthogonal matrix with entries
$U_{u,x}=N^{-1/2}(-1)^{u\cdot x}$.
Character orthogonality gives $U^TU=I$. If $\Phi$ is viewed as a vector indexed by $V$, then
$\widehat{\Phi}=N^{-1/2}U\Phi $.
For any vector $w\ne0$, every number $|w_i|/\|w\|_2$ lies in $[0,1]$. Raising it to $q\ge2$ can only decrease its square, so summing gives $\|w\|_q\le\|w\|_2$. The zero vector satisfies the same inequality. Apply this to $w=U\Phi$ and use orthogonality to obtain $\|U\Phi\|_q\le\|\Phi\|_2$. For $q>2$, H\"older's inequality with exponents $q/2$ and $q/(q-2)$ gives $\sum_x|\Phi(x)|^2\le(\sum_x|\Phi(x)|^q)^{2/q}N^{1-2/q}$. Taking square roots yields $\|\Phi\|_2\le N^{1/2-1/q}\|\Phi\|_q$. When $q=2$, this last inequality is an equality.
Therefore
\[
  \sum_{u\in V}|\widehat{\Phi}(u)|^q
  =
  N^{-q/2}\|U\Phi\|_q^q
  \le
  N^{-q/2}\|\Phi\|_2^q
  \le
  N^{-q/2}N^{q/2-1}\|\Phi\|_q^q
  =
  N^{-1}\sum_{x\in V}|\Phi(x)|^q.
\]
This proves the claim.
\end{proof}

The next lemma is the form in which the contraction is used.  It strengthens the coefficientwise estimate in Lemma~\ref{lem:ineq-nwrest-aff} by retaining the full Walsh moment.

\begin{lemma}
\label{lem:restriction-moment-contraction}
Let $F(y,z)=G_0(y)+G(y)\cdot z+h(z)$,
where $y\in\F_2^a$, $z\in\F_2^b$, $G_0:\F_2^a\to\F_2$, and
$G\colon \F_2^a\to\F_2^b$.  Then, for every even integer $q\ge2$,
\[
  M_q(F)\le M_q(h).
\]
\end{lemma}

\begin{proof}
For $u\in\F_2^a$ and $v\in\F_2^b$, we compute
\[
\begin{aligned}
  \NW_F(u,v)
  &=
  2^{-a-b}
  \sum_{y\in\F_2^a}\sum_{z\in\F_2^b}
  (-1)^{G_0(y)+G(y)\cdot z+h(z)+u\cdot y+v\cdot z} \\
  &=
  2^{-a}
  \sum_{y\in\F_2^a}
  (-1)^{G_0(y)+u\cdot y}
  \NW_h(v+G(y)).
\end{aligned}
\]
For fixed $v$, set
$\Phi_v(y)=(-1)^{G_0(y)}\NW_h(v+G(y))$.
Then
$\NW_F(u,v)=\widehat{\Phi_v}(u)$,
where the Fourier transform is taken on $\F_2^a$.  By Lemma~\ref{lem:fourier-contraction},
\[
  \sum_{u\in\F_2^a}|\NW_F(u,v)|^q
  \le
  2^{-a}\sum_{y\in\F_2^a}|\NW_h(v+G(y))|^q.
\]
Summing over $v\in\F_2^b$, we get
\[
\begin{aligned}
  M_q(F)
  &=
  \sum_{v\in\F_2^b}\sum_{u\in\F_2^a}|\NW_F(u,v)|^q \\
  &\le
  2^{-a}\sum_{y\in\F_2^a}\sum_{v\in\F_2^b}
  |\NW_h(v+G(y))|^q.
\end{aligned}
\]
For each fixed $y$, the map $v\mapsto v+G(y)$ is a bijection of $\F_2^b$.  Hence
$\sum_{v\in\F_2^b}|\NW_h(v+G(y))|^q$ equals $\sum_{v\in\F_2^b}|\NW_h(v)|^q=M_q(h)$.
Therefore
$M_q(F)\le 2^{-a}\sum_{y\in\F_2^a}M_q(h)=M_q(h)$,
as claimed.
\end{proof}

The exact hyperplane formula will follow as the case $t=1$ of Theorem~\ref{thm:Pkf}.  To treat arbitrary codimension, we need to choose the variables eliminated by Gaussian elimination with some care. The next lemma records the form of the argument from~\cite{SergeAnaPante2025} that is needed here.

\begin{lemma}\label{lem:adapted-elimination}
Let $f=f_1+\cdots+f_m$
be a direct sum of functions on pairwise disjoint sets of variables, and let $A\subseteq\F_2^n$ be an affine subspace of codimension $c$. There are affine equations defining $A$ of the form
$x_{i_j}=a_j(y)$ for $1\le j\le c$,
where the pivot variables $x_{i_1},\ldots,x_{i_c}$ are distinct and $y$ consists of the remaining variables, with the following property. Write $f=g+h$, where $g$ is the sum of the blocks containing at least one pivot variable and $h$ is the sum of the other blocks. If $y^{(1)}$ denotes the remaining variables of $g$ and $y^{(2)}$ the variables of $h$, then, after the substitutions, one can write
\begin{equation}\label{eq:adapted-elimination-form}
  f_{|A}(y^{(1)},y^{(2)})
  =G_0(y^{(1)})+G(y^{(1)})\cdot y^{(2)}+h(y^{(2)})
\end{equation}
for a Boolean function $G_0$ and a vector-valued Boolean function $G$. Moreover, at most $c$ blocks occur in $g$.
\end{lemma}

\begin{proof}
Apply Lemma~\textup{\ref{lem:ineq-nwrest-aff}} with $t=c$. Its proof gives $c$ distinct pivots and shows that each affected block has at most one substituted input whose final expression involves the untouched variables. Every other substituted input depends only on the remaining affected variables. Decomposing the block as $vB+C$ in that one input shows that its restriction is affine in the untouched variables. Summing the affected blocks gives~\eqref{eq:adapted-elimination-form}. Since each affected block contains a pivot and there are $c$ pivots, at most $c$ blocks occur in $g$.
\end{proof}

\begin{lemma}\label{lemm:monomial-direct-sum-restriction-moment}
Let $f=f_1+\cdots+f_m$
be a direct sum of monomials of degrees $r_1\ge\cdots\ge r_m\ge2$.
Let $A\subseteq\F_2^n$ be an affine subspace of codimension $c\le t<m$. If $q\ge2$ is even, then
\[
  M_q(f_{|A})
  \le
  \prod_{i=1}^{m-t}
  \left((1-2^{1-r_i})^q+(2^{r_i}-1)2^{q(1-r_i)}\right).
\]
\end{lemma}

\begin{proof}
Apply Lemma~\ref{lem:adapted-elimination}. We may write
$f_{|A}(y,z)=G_0(y)+G(y)\cdot z+h(z)$,
where $h$ is the direct sum of the monomial blocks that contain no pivot variable. Since at most $c\le t$ blocks are affected, $h$ contains at least $m-t$ of the original monomials.

By Lemma~\ref{lem:restriction-moment-contraction},
$M_q(f_{|A})\le M_q(h)$.
If $I$ is the set of indices of the unaffected blocks, then the multiplicativity of Walsh moments gives
$M_q(h)=\prod_{i\in I}A_{r_i}(q)$.
Each $A_{r_i}(q)$ is a Walsh moment, so Parseval gives $0<A_{r_i}(q)\le1$. List the unaffected indices as $i_1<\cdots<i_\ell$, where $\ell=|I|\ge m-t$. Removing the last $\ell-(m-t)$ factors can only increase their product. Moreover, $i_j\ge j$, so $r_{i_j}\le r_j$. Lemma~\textup{\ref{prop:encrease-sum}}, applied with exponent $q$, gives $A_{r_{i_j}}(q)\le A_{r_j}(q)$. Hence
\[
\prod_{i\in I}A_{r_i}(q)
\le\prod_{j=1}^{m-t}A_{r_{i_j}}(q)
\le\prod_{j=1}^{m-t}A_{r_j}(q).
\]
Combining this with $M_q(f_{\mid A})\le M_q(h)$ and substituting the definition of each $A_{r_j}(q)$ gives exactly the asserted upper bound.
\end{proof}

\begin{theorem}\label{thm:Pkf}
Let $n$ and $k$ be integers with $k\ge3$ odd. Let $f=f_1+\cdots+f_m$
be a direct sum of monomials of degrees $r_1\ge r_2\ge\cdots\ge r_m\ge2$.
Then, for every integer $0\le t<m$,
\[
  P_{k,t}(f)
  =P_k(f_1+\cdots+f_{m-t})
  =
  \frac12\left(
  1-
  \prod_{i=1}^{m-t}
  \left((1-2^{1-r_i})^{k+1}
  +(2^{r_i}-1)2^{(k+1)(1-r_i)}\right)
  \right).
\]
\end{theorem}

\begin{proof}
Put $q=k+1$. Since $k$ is odd, $q$ is even. Let $A$ be any affine subspace of codimension $c\le t$. By Lemma~\ref{lemm:monomial-direct-sum-restriction-moment},
\[
  M_q(f_{|A})
  \le
  \prod_{i=1}^{m-t}
  \left((1-2^{1-r_i})^q+(2^{r_i}-1)2^{q(1-r_i)}\right).
\]
Relation~\eqref{firstformulaof pkf} therefore gives
\[
  P_k(f_{|A})
  \ge
  \frac12\left(
  1-
  \prod_{i=1}^{m-t}
  \left((1-2^{1-r_i})^q+(2^{r_i}-1)2^{q(1-r_i)}\right)
  \right).
\]
Taking the minimum over all affine subspaces of codimension at most $t$ yields the same lower bound for $P_{k,t}(f)$.

For the reverse inequality, choose an affine subspace $A_0$ by setting one variable in each of the blocks $f_{m-t+1},\ldots,f_m$ equal to zero. These $t$ monomials vanish on $A_0$, while the other blocks remain unchanged. Hence, up to dummy variables,
$f_{|A_0}=f_1+\cdots+f_{m-t}$.
It follows that
$P_{k,t}(f)\le P_k(f_1+\cdots+f_{m-t})$.
The direct-sum formula gives
\[
  M_q(f_1+\cdots+f_{m-t})
  =
  \prod_{i=1}^{m-t}
  \left((1-2^{1-r_i})^q+(2^{r_i}-1)2^{q(1-r_i)}\right),
\]
so Relation~\eqref{firstformulaof pkf} shows that the lower bound is attained. Replacing $q$ by $k+1$ completes the proof.
\end{proof}

Taking $t=1$ gives the exact hyperplane formula directly.
\begin{corollary}\label{cor:Pkf-hyperplane}
Under the assumptions of Theorem~\textup{\ref{thm:Pkf}},
\[
P_{k,1}(f)=P_k(f_1+\cdots+f_{m-1})
=\frac12\left(1-\prod_{i=1}^{m-1}\left((1-2^{1-r_i})^{k+1}+(2^{r_i}-1)2^{(k+1)(1-r_i)}\right)\right).
\]
A minimizing hyperplane is obtained by setting one variable in the monomial $f_m$ equal to zero.
\end{corollary}
\begin{proof}
For $m\ge2$, apply Theorem~\textup{\ref{thm:Pkf}} with $t=1$. For $m=1$, the stated hyperplane makes the restriction constant, giving zero, as does the empty product in the formula.
\end{proof}

\begin{corollary}\label{cor:monomial-saturation}
Under the hypotheses of Theorem~\textup{\ref{thm:Pkf}}, $P_{k,t}(f)=0$ for $m\le t<n$. If $f$ is EA-equivalent to a direct sum of $m$ quadratic monomials, then for every $0\le t<n$,
\[
P_{k,t}(f)=\frac12\left(1-2^{-(k-1)\max(m-t,0)}\right).
\]
\end{corollary}
\begin{proof}
Choose one variable in every monomial and set all $m$ chosen variables equal to zero. Since the blocks are disjoint, these equations are independent and define a subspace $A_0$ of codimension $m$. Every monomial vanishes on $A_0$, so $P_k(f_{\mid A_0})=0$. If $m\le t<n$, this subspace is admissible in the definition of $P_{k,t}(f)$. Nonnegativity now gives $0\le P_{k,t}(f)\le P_k(f_{\mid A_0})=0$.

For a quadratic monomial, all four normalized Walsh coefficients have magnitude $1/2$. Its moment is therefore $A_2(k+1)=4(1/2)^{k+1}=2^{1-k}$. If $0\le t<m$, Theorem~\textup{\ref{thm:Pkf}} leaves $m-t$ such factors and gives $P_{k,t}(f)=\frac12(1-2^{-(k-1)(m-t)})$. If $m\le t<n$, the preceding paragraph gives zero, which equals the displayed formula because $\max(m-t,0)=0$. Proposition~\textup{\ref{prop:rnnl-equivalence}} transfers these values to every EA-equivalent function. If $m=0$, the direct sum is the zero function, and every EA-equivalent function is affine. Its restrictions are affine and have $P_k=0$, so the formula also holds in this case.
\end{proof}

\section{Threshold functions}\label{sec:threshold-functions}

We now connect restriction $k$-th-order nonhomomorphicity with restriction normalized nonlinearity and then apply this connection to threshold functions. For a Boolean function $f$ in $n$ variables and $0\le t<n$, recall that
$\rnnl_t(f)=\min_{\codim(A)\le t}\nnl(f_{\mid A})$.
The next result shows that the parameters $P_{k,t}$ recover $\rnnl_t$ as the odd order $k$ tends to infinity.

\begin{theorem}\label{thm:Pk-rnnl-limit}
Let $f$ be a Boolean function in $n$ variables and let $0\le t<n$. Put $R=1-2\rnnl_t(f)$. For every odd integer $k\ge3$,
\[
1-2\rnnl_t(f)
\le \bigl(1-2P_{k,t}(f)\bigr)^{1/(k+1)}
\le \min\left\{2^{n/(k+1)}R,\ R^{(k-1)/(k+1)}\right\}.
\]
It follows that
\[
\lim_{\substack{k\to\infty\\ k\text{ odd}}}
\bigl(1-2P_{k,t}(f)\bigr)^{1/(k+1)}
=1-2\rnnl_t(f).
\]
\end{theorem}

\begin{proof}
Put $q=k+1$. For every affine space $A$ of codimension at most $t$, Relation~\eqref{firstformulaof pkf} gives
$1-2P_k(f_{\mid A})=\sum_a |\NW_{f_{\mid A}}(a)|^q$.
If $M_A=\max_a|\NW_{f_{\mid A}}(a)|$,
then
\[
M_A^q\le \sum_a |\NW_{f_{\mid A}}(a)|^q\le 2^nM_A^q,
\]
since every restriction has at most $2^n$ Walsh coefficients. Taking the maximum over all affine spaces $A$ of codimension at most $t$ yields
\[
\max_A M_A
\le \bigl(1-2P_{k,t}(f)\bigr)^{1/q}
\le 2^{n/q}\max_A M_A.
\]
By Definition~\textup{\ref{def:nnl}},
$\max_A M_A=1-2\min_A\nnl(f_{\mid A})=1-2\rnnl_t(f)$.
Parseval also gives $\sum_a|\NW_{f_{\mid A}}(a)|^q\le M_A^{q-2}\sum_a|\NW_{f_{\mid A}}(a)|^2=M_A^{q-2}$. Since $q-2>0$ and $M_A\le R$ for every $A$, this gives $M_q(f_{\mid A})\le R^{q-2}$. Also $1-2P_{k,t}(f)=\max_A M_q(f_{\mid A})$, directly from the definition of $P_{k,t}$. Thus $(1-2P_{k,t}(f))^{1/q}\le R^{1-2/q}$, proving the second upper bound.

Finally, $R$ is independent of $q$. The first two inequalities give $0\le(1-2P_{k,t}(f))^{1/q}-R\le(2^{n/q}-1)R$. The right side tends to zero as $q\to\infty$. Along odd $k$, the exponent $q=k+1$ tends to infinity through even integers, so the squeeze theorem gives the stated limit.
\end{proof}

\begin{definition}
Let $H=x_0+V$ be an affine subspace of $\F_2^n$, with $\dim V=m$, and
let $\phi\colon \F_2^m\to H$ be an affine isomorphism.  The Walsh transform of
the restriction $F_{\mid H}$ means the Walsh transform of the Boolean
function $F\circ\phi$ on $\F_2^m$.  Different choices of $\phi$ change this
spectrum only by a permutation and signs, and therefore preserve the absolute
values of its coefficients.
\end{definition}

For $u\ne0$ and $\varepsilon\in\F_2$, write $H(u,\varepsilon)=\{x\in\F_2^n\mid u\cdot x=\varepsilon\}$. For a vector space $U$, let $U^*$ denote its dual space of linear functionals to $\F_2$.

\begin{theorem}
\label{thm:restriction-walsh-identity}
Let $F\colon \F_2^n\to\F_2$. Let $U$ be a subspace of $\F_2^n$ of
dimension $s$, let $\lambda\colon U\to\F_2$ be a linear functional, and let
\[
A(U,\lambda)=\{x\in\F_2^n\mid u\cdot x=\lambda(u)\text{ for every }u\in U\}.
\]
Then $A(U,\lambda)$ is a nonempty affine subspace of codimension $s$.
Choose $x_0\in A(U,\lambda)$, put
\[
V=U^\perp=\{x\in\F_2^n\mid u\cdot x=0\text{ for every }u\in U\},
\]
choose a linear isomorphism $T\colon \F_2^{n-s}\to V$, and define
$\phi\colon \F_2^{n-s}\to A(U,\lambda)$ by $\phi(y)=x_0+Ty$.
%
For every $a\in\F_2^{n-s}$, let $b\in\F_2^n$ satisfy
\begin{equation}\label{eq:restriction-dual-lift}
a\cdot y=b\cdot Ty\qquad\text{for every }y\in\F_2^{n-s}.
\end{equation}
Such a vector $b$ exists. Then
\begin{equation}\label{eq:restriction-walsh-general}
W_{F\circ\phi}(a)
=
(-1)^{b\cdot x_0}2^{-s}
\sum_{u\in U}(-1)^{\lambda(u)}W_F(b+u).
\end{equation}
Equivalently, in terms of the normalized Walsh transform,
\begin{equation}\label{eq:restriction-normalized-walsh-general}
\NW_{F\circ\phi}(a)
=
(-1)^{b\cdot x_0}
\sum_{u\in U}(-1)^{\lambda(u)}\NW_F(b+u).
\end{equation}
The right-hand sides of \eqref{eq:restriction-walsh-general} and
\eqref{eq:restriction-normalized-walsh-general} do not depend on the choice
of $b$ satisfying \eqref{eq:restriction-dual-lift}.

In particular, if $s=1$, $U=\langle u\rangle$ with $u\ne0$, and
$\lambda(u)=\varepsilon$, then $A(U,\lambda)=H(u,\varepsilon)$ and
\begin{equation}\label{eq:restriction-walsh-hyperplane}
W_{F\circ\phi}(a)
=
(-1)^{b\cdot x_0}\frac12
\left(W_F(b)+(-1)^\varepsilon W_F(b+u)\right).
\end{equation}
\end{theorem}

\begin{proof}
The bilinear form $(u,x)\mapsto u\cdot x$ on $\F_2^n$ is
nondegenerate. Hence the linear map
\[
\F_2^n\longrightarrow U^*,\qquad
x\longmapsto (u\mapsto u\cdot x),
\textrm{ where $U^*$ is the dual of $U$,}\]

is surjective. Therefore there exists $x_0\in\F_2^n$ such that
$u\cdot x_0=\lambda(u)$ for every $u\in U$, and thus
$A(U,\lambda)$ is nonempty. Moreover,
$A(U,\lambda)=x_0+U^\perp=x_0+V$.
Since $\dim U=s$, we have $\dim V=n-s$, and therefore
$A(U,\lambda)$ has codimension $s$.

For a fixed $a\in\F_2^{n-s}$, the map
$y\longmapsto a\cdot y$
is a linear functional on $\F_2^{n-s}$. Since $T$ is an isomorphism onto
$V$, it determines a linear functional $\ell$ on $V$ by
$\ell(Ty)=a\cdot y$. Every linear functional on a subspace of $\F_2^n$
extends to a linear functional on $\F_2^n$. By the nondegeneracy of the
dot product, such an extension has the form $x\mapsto b\cdot x$ for some
$b\in\F_2^n$. This proves the existence of $b$ satisfying
\eqref{eq:restriction-dual-lift}.

Using $x=x_0+Ty$ and \eqref{eq:restriction-dual-lift}, we obtain
\begin{align*}
W_{F\circ\phi}(a)
&=\sum_{y\in\F_2^{n-s}}(-1)^{F(x_0+Ty)+a\cdot y}
=\sum_{y\in\F_2^{n-s}}(-1)^{F(x_0+Ty)+b\cdot Ty}
=(-1)^{b\cdot x_0}
\sum_{x\in A(U,\lambda)}(-1)^{F(x)+b\cdot x}.
\end{align*}
For every $x\in\F_2^n$, the indicator function of $A(U,\lambda)$ is
\begin{equation}\label{eq:affine-space-indicator}
1_{A(U,\lambda)}(x)
=2^{-s}\sum_{u\in U}(-1)^{u\cdot x+\lambda(u)}.
\end{equation}
Indeed, if $x\in A(U,\lambda)$, then every term in the sum equals $1$.
If $x\notin A(U,\lambda)$, the map
$u\mapsto u\cdot x+\lambda(u)$ is a nonzero linear functional on $U$, so
its associated character has sum zero over $U$.

Applying \eqref{eq:affine-space-indicator}, we find
\begin{align*}
\sum_{x\in A(U,\lambda)}(-1)^{F(x)+b\cdot x}
&=\sum_{x\in\F_2^n}1_{A(U,\lambda)}(x)
  (-1)^{F(x)+b\cdot x}\\
&=2^{-s}\sum_{u\in U}(-1)^{\lambda(u)}
  \sum_{x\in\F_2^n}(-1)^{F(x)+(b+u)\cdot x}\\
&=2^{-s}\sum_{u\in U}(-1)^{\lambda(u)}W_F(b+u).
\end{align*}
This proves \eqref{eq:restriction-walsh-general}.

Since $F\circ\phi$ has $n-s$ variables, dividing
\eqref{eq:restriction-walsh-general} by $2^{n-s}$ gives
\[
\NW_{F\circ\phi}(a)
=(-1)^{b\cdot x_0}
\sum_{u\in U}(-1)^{\lambda(u)}\NW_F(b+u),
\]
which is \eqref{eq:restriction-normalized-walsh-general}.

It remains to verify independence of the choice of $b$. If $b'$ also
satisfies \eqref{eq:restriction-dual-lift}, then
$(b+b')\cdot v=0$ for every $v\in V$. Hence
$b+b'\in V^\perp=U$. Write $b'=b+u_0$ with $u_0\in U$. Since
$x_0\in A(U,\lambda)$, we have $u_0\cdot x_0=\lambda(u_0)$. Reindexing
the sum by $u'=u+u_0$ gives
\begin{align*}
&(-1)^{b'\cdot x_0}
\sum_{u\in U}(-1)^{\lambda(u)}W_F(b'+u)\\
&\quad=(-1)^{b\cdot x_0+\lambda(u_0)}
\sum_{u'\in U}(-1)^{\lambda(u'+u_0)}W_F(b+u')\\
&\quad=(-1)^{b\cdot x_0}
\sum_{u'\in U}(-1)^{\lambda(u')}W_F(b+u'),
\end{align*}
because $\lambda$ is linear. Thus the expression is independent of $b$.
The codimension-one formula follows by taking
$U=\{0,u\}$ and $\lambda(u)=\varepsilon$.
\end{proof}

\begin{corollary}\label{cor:rnnl-spectral-characterization}
Let $F\colon \F_2^n\to\F_2$ and let $0\le t<n$. Then
\begin{equation}\label{eq:rnnl-spectral-characterization}
\rnnl_t(F)
=
\frac12\left(
1-
\max_{\substack{U\le\F_2^n\\ \dim U\le t}}
\max_{\lambda\in U^*}
\max_{b\in\F_2^n}
\left|
\sum_{u\in U}(-1)^{\lambda(u)}\NW_F(b+u)
\right|
\right).
\end{equation}
\end{corollary}

\begin{proof}
Every affine subspace $A$ of codimension $s\le t$ can be written uniquely in
the form
$A=A(U,\lambda) =\{x\in\F_2^n\mid u\cdot x=\lambda(u)\text{ for every }u\in U\}$,
where $U$ is the orthogonal complement of the direction space of $A$ and
$\lambda(u)=u\cdot x_0$ for any $x_0\in A$. Here $\dim U=s$.
Conversely, every pair consisting of a subspace $U\le\F_2^n$ and a linear
functional $\lambda\in U^*$ defines a nonempty affine subspace
$A(U,\lambda)$ of codimension $\dim U$ by
Theorem~\ref{thm:restriction-walsh-identity}.

Fix such a pair $(U,\lambda)$ and an affine parametrization $\phi$ of
$A(U,\lambda)$ as in the theorem. For every
$a\in\F_2^{n-s}$, Equation~\eqref{eq:restriction-normalized-walsh-general}
gives a vector $b\in\F_2^n$ such that
\[
|\NW_{F\circ\phi}(a)|
=
\left|
\sum_{u\in U}(-1)^{\lambda(u)}\NW_F(b+u)
\right|.
\]
Conversely, for every $b\in\F_2^n$, the map
$y\mapsto b\cdot Ty$ is a linear functional on $\F_2^{n-s}$, and hence
there is a unique $a\in\F_2^{n-s}$ such that
$a\cdot y=b\cdot Ty$ for every $y$. Therefore
\begin{equation}\label{eq:max-restricted-spectrum}
\max_a|\NW_{F\circ\phi}(a)|
=
\max_{b\in\F_2^n}
\left|
\sum_{u\in U}(-1)^{\lambda(u)}\NW_F(b+u)
\right|.
\end{equation}
By the definition of normalized nonlinearity,
\[
\nnl(F_{\mid A(U,\lambda)})
=
\frac12\left(1-
\max_a|\NW_{F\circ\phi}(a)|\right).
\]
Taking the minimum over all affine subspaces of codimension at most $t$ and
using \eqref{eq:max-restricted-spectrum} proves
\eqref{eq:rnnl-spectral-characterization}.
\end{proof}
\begin{remark}
    An equivalent formulation of Theorem~\ref{thm:restriction-walsh-identity} is as follows. Consider an affine space $U$ of $\F_2^n$, of codimension $k$ (that is, dimension $n-k)$, defined by the $k$ linearly independent equations $u_1\cdot x = \varepsilon_1, \ldots, u_k\cdot x = \varepsilon_k$, with $u_i\in \F_2^n$ and $\varepsilon_i\in \F_2$. The indicator function of $U$ in $\F_2^n$ is
    \[
    1_U(x) = \frac{1}{2^k}\prod_{i=1}^k \left(1+(-1)^{u_i\cdot x + \varepsilon_i}\right) = \frac{1}{2^k}\left( \sum_{c_1,\ldots,c_k\in \F_2} (-1)^{(c_1 u_1+\cdots+ c_ku_k)\cdot x + c_1\varepsilon_1+\cdots+c_k\varepsilon_k} \right).
    \]
    Consider an affine isomorphism $\phi\colon \F_2^{n-k} \rightarrow U$ defined by $\phi(y) = b_0 + Ty$, with $T$ being a $n \times (n-k)$ matrix of rank $n-k$ and $b_0\in U$ (vectors in $\F_2^n$ being considered column vectors). For each $a\in \F_2^{n-k}$, 
    choose a left inverse $S$ of $T$, so that $ST=I_{n-k}$, and set $b=S^Ta$. Then $T^Tb=a$, and hence $a\cdot y=b\cdot Ty$ for all $y\in\F_2^{n-k}$. We have
\begin{eqnarray*}
    W_{F\circ \phi}(a) &=&
        \sum_{y\in \F_2^{n-k}} (-1)^{F(\phi(y))+a\cdot y}\\ 
        &=& \sum_{y\in \F_2^{n-k}} (-1)^{F(\phi(y))+b\cdot Ty}\\
&=& \sum_{y\in \F_2^{n-k}} (-1)^{F(\phi(y))+b\cdot (\phi(y)+b_0)}\\  
        &=& (-1)^{b \cdot b_0}\sum_{x\in U} (-1)^{F(x)+b\cdot x}\\
        &=& (-1)^{b \cdot b_0}\sum_{x\in \F_2^n} 1_U(x)(-1)^{F(x)+b\cdot x}\\
        &=&  (-1)^{b \cdot b_0} \sum_{x\in \F_2^n} \frac{1}{2^k}\left( \sum_{c_1,\ldots,c_k\in \F_2} (-1)^{F(x)+(b+c_1 u_1+\cdots+ c_ku_k)\cdot x + c_1\varepsilon_1+\cdots+c_k\varepsilon_k} \right)\\ 
        &=&  (-1)^{b \cdot b_0}   \frac{1}{2^k}\left( \sum_{c_1,\ldots,c_k\in \F_2} (-1)^{ c_1\varepsilon_1+\cdots+c_k\varepsilon_k}  W_F(b+c_1 u_1+\cdots+ c_ku_k)\right).
    \end{eqnarray*}
Dividing by $2^{n-k}$ for normalization, we obtain
\[
  \NW_{F\circ \phi}(a) =  (-1)^{b \cdot b_0}  \left( \sum_{c_1,\ldots,c_k\in \F_2} (-1)^{ c_1\varepsilon_1+\cdots+c_k\varepsilon_k}  \NW_F(b+c_1 u_1+\cdots+ c_ku_k)\right).
\]
For the particular case of $k=1$ this becomes
\[
  \NW_{F\circ \phi}(a) =  (-1)^{b \cdot b_0}  \left( \NW_F(b) + (-1)^{\varepsilon_1}\NW_F(b+u_1)
  \right).
\]
\end{remark}

A Boolean function is symmetric if its value depends only on the Hamming weight of its input. For a vector $x$, its support $\supp(x)$ is the set of nonzero coordinate indices. Coordinate permutations preserve the Walsh transform of a symmetric function, so its coefficients also depend only on Hamming weight.

The preceding Theorem~\ref{thm:restriction-walsh-identity} has a particularly simple form for symmetric functions.

\begin{proposition}
\label{prop:symmetric-hyperplane-restrictions}
Let $F\colon \F_2^n\to\F_2$ be a symmetric Boolean function. For $0\le r\le n$, let
$\Omega_r$ denote the common value of $\NW_F(b)$ over all
$b\in\F_2^n$ with $\w(b)=r$. Let
\[
H(u,\varepsilon)=\{x\in\F_2^n\mid u\cdot x=\varepsilon\},
\qquad \w(u)=m,
\]
where $u\ne 0$ and $\varepsilon\in\F_2$. Then
\begin{equation}\label{eq:symmetric-hyperplane-restriction}
\nnl(F_{\mid H(u,\varepsilon)})
=
\frac12\left(
1-
\max_{\substack{0\le p\le m\\0\le q\le n-m}}
\left|
\Omega_{p+q}+(-1)^\varepsilon\Omega_{m-p+q}
\right|
\right).
\end{equation}
In particular, the nonlinearity of the restriction depends only on
$\w(u)$ and $\varepsilon$, and not on the particular choice of $u$.
For a coordinate hyperplane,
\begin{equation}\label{eq:symmetric-coordinate-restriction}
\nnl(F_{\mid\{x_i=\varepsilon\}})
=
\frac12\left(
1-
\max_{0\le r\le n-1}
\left|
\Omega_r+(-1)^\varepsilon\Omega_{r+1}
\right|
\right).
\end{equation}
\end{proposition}

\begin{proof}
Fix an affine parametrization of $H(u,\varepsilon)$. By the
codimension-one case of Theorem~\ref{thm:restriction-walsh-identity},
every normalized Walsh coefficient of the restriction, up to a factor of
absolute value one, has the form
$\NW_F(b)+(-1)^\varepsilon\NW_F(b+u)$
for some $b\in\F_2^n$. Conversely, every $b\in\F_2^n$ gives a Walsh
coefficient of the restriction, by the converse part of the argument in the
proof of Corollary~\ref{cor:rnnl-spectral-characterization}.

For a fixed $b$, set
$p=|\supp(b)\cap\supp(u)|$ and $q=|\supp(b)\setminus\supp(u)|$.
Then $0\le p\le m$, $0\le q\le n-m$, and
$\w(b)=p+q$ and $\w(b+u)=m-p+q$.
Since $F$ is symmetric, its normalized Walsh transform depends only on
Hamming weight. Hence
$\NW_F(b)=\Omega_{p+q}$ and $\NW_F(b+u)=\Omega_{m-p+q}$.
Conversely, for every pair $(p,q)$ with $0\le p\le m$ and
$0\le q\le n-m$, one can choose $b$ having exactly $p$ nonzero
coordinates in $\supp(u)$ and exactly $q$ nonzero coordinates outside
$\supp(u)$. Therefore
\[
\max_a|\NW_{F_{\mid H(u,\varepsilon)}}(a)|
=
\max_{\substack{0\le p\le m\\0\le q\le n-m}}
\left|
\Omega_{p+q}+(-1)^\varepsilon\Omega_{m-p+q}
\right|.
\]
Using the definition of normalized nonlinearity gives
\eqref{eq:symmetric-hyperplane-restriction}.

If $H(u,\varepsilon)$ is a coordinate hyperplane, then $m=1$. For
$p=0$ the two indices are $q$ and $q+1$, while for $p=1$ they are
$q+1$ and $q$. Thus the maximum reduces to the one in
\eqref{eq:symmetric-coordinate-restriction}.
\end{proof}

\begin{remark}
\label{rem:symmetric-coordinate-differences}
Coordinate restrictions preserve symmetry. Thus, if
$(\Omega_0,\ldots,\Omega_n)$ denotes the weight-indexed normalized Walsh
spectrum of a symmetric Boolean function, restriction to $x_i=0$
corresponds to taking adjacent sums
$\Omega_r+\Omega_{r+1}$, whereas restriction to $x_i=1$
corresponds to taking adjacent differences
$\Omega_r-\Omega_{r+1}$.
This provides a simple interpretation of
Proposition~\textup{\ref{prop:symmetric-hyperplane-restrictions}} for coordinate hyperplanes.
\end{remark}

\begin{theorem}
\label{thm:B12}
Let $F\colon \F_2^n\to\F_2$, with $n\ge 1$. Let $B_1$ and $B_2$
be the two largest elements, counted at distinct frequencies, of the multiset
$\bigl\{|\NW_F(b)|\mid b\in\F_2^n\bigr\}$. Equivalently,
\[
B_1+B_2
=
\max_{\substack{b_1,b_2\in\F_2^n\\ b_1\ne b_2}}
\left(|\NW_F(b_1)|+|\NW_F(b_2)|\right),
\]
where $B_1\ge B_2$.
Then
\[
\max_H\max_{a\in\F_2^{n-1}}
 \bigl|\NW_{F_{\mid H}}(a)\bigr|
 =B_1+B_2,
\]
where $H$ ranges over all affine hyperplanes of $\F_2^n$, and the
Walsh transform of $F_{\mid H}$ is computed through any affine
parametrization of $H$. It follows that
\begin{equation}
    \label{eq:B12}
    \rnnl_1(F)=\frac12\bigl(1-B_1-B_2\bigr),
\end{equation}
A hyperplane $H(u,\varepsilon)$ attains this minimum if and only if $u$ can be written as a sum $u=b'_1+b'_2$ with $b'_1\neq b'_2$, $B_1 = |\NW_F(b'_1)|$ and $B_2 = |\NW_F(b'_2)|$ and $\varepsilon$ satisfies the relation
\begin{equation}
    \label{eq:B12-epsilon}
    \varepsilon=
\begin{cases}
0,&\NW_F(b'_1)\NW_F(b'_2)>0,\\
1,&\NW_F(b'_1)\NW_F(b'_2)<0,\\
0\text{ or }1,&\NW_F(b'_1)\NW_F(b'_2)=0.
\end{cases}
\end{equation}
Moreover,
$\rnnl_1(F) \le \nnl(F)$,
with equality if and only if $F$ is an affine function.
\end{theorem}

\begin{proof}
Let $H(u,\varepsilon)$ be an affine hyperplane, with $u\ne 0$.
Using the codimension-one case of Theorem~\ref{thm:restriction-walsh-identity} and choosing an affine parametrization $\phi(y)=x_0+Ty$ of $H(u,\varepsilon)$, we have that for
every $a\in\F_2^{n-1}$ there exists $b\in\F_2^n$ such that
\begin{equation}\label{eq:proof_B12-1}
 \bigl|\NW_{F\circ\phi}(a)\bigr|
=
\left|
\NW_F(b)+(-1)^\varepsilon\NW_F(b+u)
\right|.   
\end{equation}
Note that for any $b\in \F_2^n$, there exists an $a\in \F_2^{n-1}$ such that~\eqref{eq:proof_B12-1} holds. Namely, let $a\in\F_2^{n-1}$ be the unique vector
satisfying
$a\cdot y=b\cdot Ty$ for all $y\in\F_2^{n-1}$. Such a value of $a$ always exists as $y\mapsto b\cdot Ty$ is a linear
functional on $\F_2^{n-1}$. 

Since $u\ne0$, the frequencies $b$ and $b+u$ on~\eqref{eq:proof_B12-1} are distinct. Hence
\begin{equation}
    \label{eq:proof-B12-2}
    \bigl|\NW_{F\circ\phi}(a)\bigr|
\le |\NW_F(b)|+|\NW_F(b+u)|
\le B_1+B_2.
\end{equation}
Equality holds in both inequalities in~\eqref{eq:proof-B12-2} precisely when the two magnitudes $|\NW_F(b)|$ and $|\NW_F(b+u)|$ are $B_1$ and $B_2$ in either order, and $(-1)^\varepsilon\NW_F(b)\NW_F(b+u)\ge0$.
The first condition is equivalent to the existence of distinct frequencies
$b'_1,b'_2$ such that $u=b'_1+b'_2$, $|\NW_F(b'_1)|=B_1$, and $|\NW_F(b'_2)|=B_2$,
where the roles of $b'_1$ and $b'_2$ may be interchanged. The second condition
is equivalent to~\eqref{eq:B12-epsilon}. Therefore these are exactly the
hyperplanes for which the restriction has minimum nonlinearity.

Choose distinct frequencies $b_1,b_2$ attaining $B_1,B_2$. Set $u=b_1+b_2$, and choose $\varepsilon$ according to~\eqref{eq:B12-epsilon} so that the two coefficients have matching signs. If a coefficient is zero, either value works. This constructs a hyperplane attaining both inequalities in~\eqref{eq:proof-B12-2}. This proves that
\[
\max_H\max_{a\in \F_2^{n-1}}\bigl|\NW_{F_{\mid H}}(a)\bigr|
= B_1+B_2.
\]


The full space has largest absolute normalized coefficient $B_1$, whereas the best hyperplane has $B_1+B_2\ge B_1$. Thus including the full space does not change the restriction minimum. To prove~\eqref{eq:B12},
\[
\begin{aligned}
\rnnl_1(F)
&=\min_H\nnl(F_{\mid H})
=\min_H\frac12
 \left(1-\max_a|\NW_{F_{\mid H}}(a)|\right)\\
&=\frac12\left(
1-\max_H\max_a|\NW_{F_{\mid H}}(a)|
\right)
=\frac12(1-B_1-B_2).
\end{aligned}
\]
Finally,
\[
\rnnl_1(F)=\frac12(1-B_1-B_2)
\le \frac12(1-B_1)=\nnl(F).
\]
Equality holds if and only if $B_2=0$, that is, if and only if the Walsh
spectrum of $F$ has exactly one nonzero coefficient. Suppose this coefficient
is $W_F(b_0)$. Parseval's identity gives
$W_F(b_0)^2=2^{2n}$,
so $W_F(b_0)=(-1)^c2^n$ for some $c\in\F_2$. Walsh inversion then yields
\[
(-1)^{F(x)}
=2^{-n}W_F(b_0)(-1)^{b_0\cdot x}
=(-1)^{b_0\cdot x+c}
\]
for every $x\in\F_2^n$. Hence $F(x)=b_0\cdot x+c$, so $F$ is affine.
Conversely, an affine function has exactly one nonzero Walsh coefficient, and
therefore $B_2=0$.
\end{proof}
\begin{theorem}\label{thm:strict-nonhomomorphicity}
Let $F\colon\F_2^n\to\F_2$, where $n\ge2$, and let $k\ge3$ be odd. Then $P_{k,1}(F)=P_k(F)$ if and only if $F$ is affine. More explicitly, for any $u\ne0$, put $q=k+1$, and choose a set $R_u$ containing one representative of each coset of $\langle u\rangle$ in $\F_2^n$. With $a_b=\NW_F(b)$ and $c_b=\NW_F(b+u)$, one has
\begin{equation}\label{eq:hyperplane-moment-average}
\frac{M_q(F_{\mid H(u,0)})+M_q(F_{\mid H(u,1)})}{2}
=M_q(F)+\sum_{b\in R_u}\sum_{j=1}^{q/2-1}
\binom q{2j}a_b^{q-2j}c_b^{2j}.
\end{equation}
\end{theorem}
\begin{proof}
The map taking an ambient frequency $b$ to the corresponding frequency on the direction space $u^\perp$ has kernel $\langle u\rangle$. Thus the hyperplane restriction formula gives, for each $\varepsilon$,
\[
M_q(F_{\mid H(u,\varepsilon)})
=\sum_{b\in R_u}|a_b+(-1)^\varepsilon c_b|^q.
\]
Since $q$ is even, averaging over $\varepsilon$ and expanding by the binomial theorem leaves exactly the even powers of $c_b$. The terms with $j=0$ and $j=q/2$ sum to $M_q(F)$, proving~\eqref{eq:hyperplane-moment-average}.

If $F$ is not affine, Walsh inversion and Parseval imply that its spectrum has two distinct nonzero coefficients. Choose their frequencies $b_1,b_2$ and put $u=b_1+b_2$. Let $b\in R_u$ represent their common coset. Then $a_b$ and $c_b$ are precisely these two nonzero coefficients, in either order. The term with $j=1$ in~\eqref{eq:hyperplane-moment-average} is $\binom q2 a_b^{q-2}c_b^2>0$, since $q\ge4$ and both exponents are positive and even. Every other term is nonnegative for the same parity reason. Hence the average of the two restricted moments is strictly greater than $M_q(F)$. At least one of them, say the one with right side $\varepsilon$, must satisfy $M_q(F_{\mid H(u,\varepsilon)})>M_q(F)$. Subtracting from one and dividing by two gives $P_k(F_{\mid H(u,\varepsilon)})<P_k(F)$. If $F$ is affine, both $P_k(F)$ and all its restricted values are zero.
\end{proof}

\begin{corollary}\label{cor:strict-restriction-sequence}
Let $k\ge3$ be odd and $0\le t<n-1$. If $P_{k,t}(F)>0$, then $P_{k,t+1}(F)<P_{k,t}(F)$.
\end{corollary}
\begin{proof}
The family of admissible affine subspaces is finite, so choose $A$ with $P_k(F_{\mid A})=P_{k,t}(F)>0$. An affine function has $P_k=0$, hence $F_{\mid A}$ is not affine. Since $\codim A\le t<n-1$, its dimension is at least two. Theorem~\textup{\ref{thm:strict-nonhomomorphicity}}, applied in affine coordinates on $A$, gives a hyperplane $B$ of $A$ with $P_k((F_{\mid A})_{\mid B})<P_k(F_{\mid A})$.

The space $B$ is an ambient affine subspace, with $\codim B=\codim A+1\le t+1$. Restricting first to $A$ and then to $B$ gives the same function as restricting directly to $B$, up to an affine change of coordinates, which preserves $P_k$. Therefore $P_{k,t+1}(F)\le P_k(F_{\mid B})<P_k(F_{\mid A})=P_{k,t}(F)$, as required.
\end{proof}

\begin{remark}
    The last statement in Theorem~\textup{\ref{thm:B12}} above answers negatively a question in~\textup{\cite{SergeAnaPante2025}}, regarding the existence of any functions $F$, other than affine functions, such that  $\rnnl_1(F) = \nnl(F)$.
\end{remark}
\begin{example}
For a direct sum of monomials with $r_1\ge\cdots\ge r_m\ge2$, put $\alpha_i=1-2^{1-r_i}$ and $\beta_i=2^{1-r_i}$. The largest absolute normalized Walsh coefficient is $B_1=\prod_i\alpha_i$. The second largest, counted at a distinct frequency, is $B_2=\beta_m\prod_{i<m}\alpha_i$. Indeed, each nonzero block frequency replaces $\alpha_i$ by $\beta_i$, and $\beta_i/\alpha_i\le1$ is largest for a smallest-degree block. This argument includes quadratic blocks, for which the ratio is $1$. Theorem~\textup{\ref{thm:B12}} gives
\[
\rnnl_1(f)=\frac12\left(1-\prod_{i=1}^{m-1}(1-2^{1-r_i})\right).
\]
This recovers the direct-sum hyperplane formula of~\cite{SergeAnaPante2025}. For $m=1$, the empty product is $1$.
\end{example}
We next determine $\rnnl_1$ exactly for every threshold function. This settles the two conjectures stated in \cite[Conjectures~1 and~2]{SergeAnaPante2025}.

\begin{definition}\label{def:thresholds-funct}
Let $n\ge1$ and $d$ be integers with $0\le d\le n+1$. The threshold function $T_{d,n}$ is defined by
\[
T_{d,n}(x)=
\begin{cases}
0,&\text{if }\w(x)<d,\\
1,&\text{if }\w(x)\ge d.
\end{cases}
\]
When $n$ is odd and $d=(n+1)/2$, this is the majority function in $n$ variables.
\end{definition}

For integers $0\le q\le m$ and $0\le j\le m$, let
\[
K_j^{(m)}(q)=\sum_{h=0}^{j}(-1)^h
\binom{q}{h}\binom{m-q}{j-h}
\]
be the binary Krawtchouk polynomial, with the convention that $\binom ab=0$ when $b<0$ or $b>a$.
It is known (see, for example,~\cite[Section~6.1]{cusickStanica2017})
that the Walsh coefficient of any symmetric function $F$ at any vector $a$ of weight $r$ can be expressed using Krawtchouk polynomials by the formula
\begin{equation}\label{eq:Kraw-symm-Walsh}
    W_F(a) = \sum_{j=0}^{n}(-1)^{F(j)}K_j^{(n)}(r),
\end{equation}
where, by abuse of notation, we denoted by $F(j)$ the value of $F$ at any vector of Hamming weight $j$.
\begin{lemma}\label{lem:threshold-walsh-krawtchouk}
Let $1\le d\le n$ and let $a\in\F_2^n$ have Hamming weight $r$. Then
\[
W_{T_{d,n}}(0)=2\sum_{j=0}^{d-1}\binom nj-2^n,
\]
and, for $1\le r\le n$,
\[
W_{T_{d,n}}(a)=2K_{d-1}^{(n-1)}(r-1).
\]
\end{lemma}

\begin{proof} 
Using~\eqref{eq:Kraw-symm-Walsh}, the Walsh coefficient of the threshold function $T_{d,n}$ at a vector of weight $r$ is
the difference between $\sum_{j=0}^{d-1}K_j^{(n)}(r)$ and $\sum_{j=d}^{n}K_j^{(n)}(r)$.
For $r=0$, this gives the first formula. If $r>0$, the complete Krawtchouk sum $\sum_{j=0}^{n}K_j^{(n)}(r)$ is zero, and hence
\[
W_{T_{d,n}}(a)=2\sum_{j=0}^{d-1}K_j^{(n)}(r).
\]
Using
\[
K_j^{(n)}(r)=K_j^{(n-1)}(r-1)-K_{j-1}^{(n-1)}(r-1),
\]
where $K_{-1}^{(n-1)}=0$, the sum telescopes to $K_{d-1}^{(n-1)}(r-1)$.
\end{proof}
Carlet et al.~\cite{CSA2023,CSA2025} studied the largest Walsh coefficients of threshold functions. Away from the middle thresholds, the zero-frequency coefficient has largest magnitude. At the middle thresholds, the coefficients of weights $1$ and $n$ attain the maximum. We will prove these statements together with the strict bounds at all other nonzero frequencies.

In order to apply Theorem~\ref{thm:B12} to threshold functions, we need to determine the two largest absolute values of the  Walsh coefficients (considered with multiplicities). The following result incorporates the results from~\cite{CSA2023} mentioned above, and additionally covers the second largest value in the cases where the largest value is only attained once. The proof below treats every case and identifies the strict spectral inequalities needed for the minimizing-hyperplane classification.
\begin{proposition} \label{prop:threshold-B12}
    Let $T_{d,n}$ be a threshold function with $1<d<n$. Let  $a\in \F_2^n$ of weight $r$. If $r=1$, then $W_{T_{d,n}}(a) = 2 \binom{n-1}{d-1}$, whereas if $r=n$, then $W_{T_{d,n}}(a) = 2 (-1)^{d-1}\binom{n-1}{d-1}$.  If $1<r<n$ we have $|W_{T_{d,n}}(a)| <2 \binom{n-1}{d-1}$. Moreover,
    \begin{enumerate}
        \item If $1<d<\frac{n}{2}$, then $W_{T_{d,n}}(0) = 2\sum_{j=0}^{d-1}\binom{n}{j}-2^{n} <0$ and $|W_{T_{d,n}}(0) |>2\binom{n-1}{d-1}$.
        \item If $d=\frac{n}{2}$, then $W_{T_{d,n}}(0) = -2\binom{n-1}{d-1}$,        \item If $d=\frac{n+1}{2}$, then $W_{T_{d,n}}(0) =0$.
        \item If $d=\frac{n}{2}+1$, then $W_{T_{d,n}}(0) = 2\binom{n-1}{d-1}$.        \item If $n>d>\frac{n}{2}+1$, then $W_{T_{d,n}}(0) = 2\sum_{j=0}^{d-1}\binom{n}{j}-2^{n} >2\binom{n-1}{d-1}$.
    \end{enumerate}
\end{proposition}
\begin{proof}
Write $C=\binom{n-1}{d-1}$. The coefficients at weights $1$ and $n$ follow from Lemma~\textup{\ref{lem:threshold-walsh-krawtchouk}} by setting the Krawtchouk argument to $0$ and $n-1$.

For $1<r<n$, let $m=n-1$, $s=d-1$, and choose a vector $v\in\F_2^m$ of weight $r-1$. Then
\[
K_s^{(m)}(r-1)=\sum_{\substack{x\in\F_2^m\\\w(x)=s}}(-1)^{v\cdot x}.
\]
Here $1\le\w(v)\le m-1$ and $1\le s\le m-1$. Choose an index $i$ in $\supp(v)$ and an index $j$ outside it. Choose $s-1$ further indices among the remaining $m-2$ indices. The two weight-$s$ vectors obtained by adding $i$ or $j$ have opposite character signs. Thus both signs occur in the sum and its absolute value is strictly smaller than $\binom ms=C$.

Put $A_0=W_{T_{d,n}}(0)$. If $d<n/2$, binomial symmetry gives
$A_0=-\sum_{j=d}^{n-d}\binom nj$.
The sum contains the two distinct endpoint terms. Hence $|A_0|\ge2\binom nd>2C$, since $\binom nd=(n/d)C>C$. If $n$ is even and $d=n/2$, the same sum consists of the central term and gives $A_0=-\binom n{n/2}=-2C$. If $n$ is odd and $d=(n+1)/2$, binomial symmetry gives $A_0=0$.

For the remaining cases, put $d'=n-d+1$. Binomial symmetry in the zero-frequency formula gives $W_{T_{d',n}}(0)=-A_0$, while $\binom{n-1}{d'-1}=C$. Applying the two lower-threshold cases just proved yields $A_0=2C$ at $d=n/2+1$ and $A_0>2C$ above it.
\end{proof}

\begin{theorem}\label{thm:threshold-main}
Let $n\ge3$ and $1< d<n$. Then
\[
\rnnl_1(T_{d,n})
=\nnl(T_{d,n})-\frac{1}{2^n}\binom{n-1}{d-1}.
\]
Equivalently,
\[
\rnnl_1(T_{d,n})=
\begin{cases}
\displaystyle \frac{1}{2^{n-1}}\sum_{j=0}^{d-2}\binom{n-1}{j},
& d<\dfrac{n+1}{2},\\[8pt]
\displaystyle \frac12-\frac{1}{2^{n-1}}\binom{n-1}{\frac{n-1}{2}},
& d=\dfrac{n+1}{2},\\[8pt]
\displaystyle \frac{1}{2^{n-1}}\sum_{j=d}^{n-1}\binom{n-1}{j},
& d>\dfrac{n+1}{2}.
\end{cases}
\]
The middle case occurs only when $n$ is odd. Moreover, the restriction to a hyperplane $H$  achieves the minimum nonlinearity $\nnl((T_{d,n})_{|H}) = \rnnl_1(T_{d,n})$ if and only if $H$ is defined, up to a permutation of the variables, by one of the following equations.
\begin{itemize}
\item[$(i)$] If $1<d<\frac{n}{2}$,  the equation $x_1=1$ or the equation $x_1+ x_2 + \cdots + x_n = (d\bmod 2)$.
\item[$(ii)$] If $d=\frac{n}{2}$,  the equation $x_1=1$ or the equation $x_1 + x_2 = 0$ or  the equation  $x_1+ x_2 + \cdots + x_{n-1} = ((d-1) \bmod 2)$ or  the equation  $x_1+ x_2 + \cdots + x_n = (d\bmod 2)$.
\item[$(iii)$] If $d=\frac{n+1}{2}$,   the equation $x_1 +x_2 = 0$ or  the equation  $x_1+ x_2 + \cdots + x_{n-1} = ((d-1) \bmod 2)$ .
\item[$(iv)$] If $d=\frac{n}{2}+1$,  the equation $x_1=0$ or the equation $x_1 + x_2 = 0$ or  the equation  $x_1+ x_2 + \cdots + x_{n-1} = ((d-1) \bmod 2)$ or  the equation  $x_1+ x_2 + \cdots + x_n = ((d-1)\bmod 2)$.
\item[$(v)$] If $n>d>\frac{n}{2}+1$,  the equation $x_1=0$ or the equation $x_1+ x_2 + \cdots + x_n = ((d-1)\bmod 2)$. 
\end{itemize}

\end{theorem}
\begin{proof}
Put $C=\binom{n-1}{d-1}$, $L=2C$, $A_0=W_{T_{d,n}}(0)$, and $\eta=(-1)^{d-1}$. Proposition~\textup{\ref{prop:threshold-B12}} shows that all weight-$1$ coefficients are $L$, the weight-$n$ coefficient is $\eta L$, and every other nonzero-frequency coefficient has magnitude strictly smaller than $L$.

Outside the odd-majority case, $|A_0|\ge L$. Thus the two largest absolute normalized coefficients satisfy
$B_1=2^{-n}|A_0|$, $B_2=2^{-n}L$, and $\nnl(T_{d,n})=\frac12(1-B_1)$.
These formulas hold also when $|A_0|=L$. Theorem~\textup{\ref{thm:B12}} gives $\rnnl_1(T_{d,n})=\nnl(T_{d,n})-2^{-n}C$.
In the odd-majority case, $A_0=0$ and $B_1=B_2=2^{-n}L$. Since $\nnl(T_{d,n})=\frac12(1-2^{-n}L)$, the same unified formula follows.

If $d<(n+1)/2$, then $\nnl(T_{d,n})=2^{-n}\sum_{j=0}^{d-1}\binom nj$. Pascal's identity gives
\[
\sum_{j=0}^{d-1}\binom nj-C
=2\sum_{j=0}^{d-2}\binom{n-1}j.
\]
If $d>(n+1)/2$, then $\nnl(T_{d,n})=2^{-n}\sum_{j=d}^n\binom nj$, and Pascal's identity gives
\[
\sum_{j=d}^n\binom nj-C
=2\sum_{j=d}^{n-1}\binom{n-1}j.
\]
For odd majority the formula is $\frac12-2^{1-n}C$. These are the three displayed cases.

We now classify every minimizing hyperplane using the sign criterion in Theorem~\textup{\ref{thm:B12}}. Let $e_i$ be the $i$-th standard basis vector and let $\mathbf1$ be the all-ones vector. Up to permutation of the variables, sums of maximizing frequency pairs have the following normals and sign conditions.
\[
\begin{array}{c|c|c}
\text{frequency pair}&\text{normal}&\text{required }\varepsilon\\ \hline
0,e_i&e_i&1\text{ if }A_0<0,\ 0\text{ if }A_0>0\\
0,\mathbf1&\mathbf1&d\bmod2\text{ if }A_0<0,\ (d-1)\bmod2\text{ if }A_0>0\\
e_i,e_j\ (i\ne j)&e_i+e_j&0\\
e_i,\mathbf1&\mathbf1+e_i&(d-1)\bmod2
\end{array}
\]
The last two rows use the signs $L,L$ and $L,\eta L$, respectively. The first two rows use the sign of $A_0$ and $\eta=(-1)^{d-1}$.

If $1<d<n/2$, the zero frequency is the unique frequency of largest magnitude. The second frequency must be $e_i$ or $\mathbf1$, and $A_0<0$. The first two rows give exactly item $(i)$.

If $d=n/2$, the maximizing frequencies are $0,e_1,\ldots,e_n,\mathbf1$, all with magnitude $L$, and $A_0<0$. Every pair of distinct maximizing frequencies is one of the four rows. Their normals and signs give exactly item $(ii)$.

If $d=(n+1)/2$, the zero coefficient vanishes. The maximizing frequencies are $e_1,\ldots,e_n,\mathbf1$. Only the last two rows are available, giving exactly item $(iii)$.

If $d=n/2+1$, the maximizing frequencies are again $0,e_1,\ldots,e_n,\mathbf1$, but $A_0>0$. All four rows apply with the positive zero-frequency sign and give exactly item $(iv)$.

If $n>d>n/2+1$, the zero frequency is uniquely largest and $A_0>0$. The first two rows give exactly item $(v)$. Strictness of the remaining spectral bounds excludes every other frequency pair and hence every other hyperplane. When two listed equations coincide at a small value of $n$, they describe the same hyperplane.
\end{proof}

\begin{remark}\label{rem:threshold-endpoints}
For $n\ge2$, the unified formula in Theorem~\textup{\ref{thm:threshold-main}} also holds for $d=1$ and $d=n$, and its value is zero. Indeed, $T_{n,n}$ is a monomial with $\nnl(T_{n,n})=2^{-n}$, and it vanishes on $x_1=0$. Also $T_{1,n}(x)=1+T_{n,n}(\mathbf1+x)$, so it has the same normalized nonlinearity and restriction minimum. The endpoint functions $T_{0,n}$ and $T_{n+1,n}$ are constant and have restriction minimum zero. The five-case classification is asserted only for $1<d<n$.
\end{remark}

The next two corollaries of Theorem~\ref{thm:threshold-main} state explicitly that the conjectures in \cite{SergeAnaPante2025} are resolved.

\begin{corollary}\label{cor:CCDS-threshold-conjecture}
Let $n\ge3$ and $1\le d\le n$. For $d>(n+1)/2$,
\[
\rnnl_1(T_{d,n})=\nnl(T_{d,n-1}),
\]
and for $d<(n+1)/2$,
\[
\rnnl_1(T_{d,n})=\nnl(T_{d-1,n-1}).
\]
Thus \textup{\cite[Conjecture~1]{SergeAnaPante2025}} holds.
\end{corollary}

\begin{proof}
Suppose first that $1<d<n$ and $d>(n+1)/2$. Items $(iv)$ and $(v)$ of Theorem~\textup{\ref{thm:threshold-main}} include $x_1=0$ among the minimizing hyperplanes. On this hyperplane, $T_{d,n}(0,x_2,\ldots,x_n)=T_{d,n-1}(x_2,\ldots,x_n)$. Its normalized nonlinearity is therefore $\nnl(T_{d,n-1})$, and minimality gives the first equality.

If $1<d<n$ and $d<(n+1)/2$, items $(i)$ and $(ii)$ include $x_1=1$ among the minimizing hyperplanes. Since the fixed coordinate contributes one to the weight, $T_{d,n}(1,x_2,\ldots,x_n)=T_{d-1,n-1}(x_2,\ldots,x_n)$. Minimality gives the second equality.

At $d=n$, the restriction to $x_1=0$ is the constant function $T_{n,n-1}=0$, and at $d=1$, the restriction to $x_1=1$ is $T_{0,n-1}=1$. Both have normalized nonlinearity zero. Nonnegativity of normalized nonlinearity gives $\rnnl_1(T_{d,n})=0$ at these endpoints, so the same two equalities hold there.
\end{proof}

\begin{corollary}\label{cor:CCDS-majority-conjecture}
Let $n\ge3$ be odd and $d=(n+1)/2$. Then the minimum defining $\rnnl_1(T_{d,n})$ is attained on the hyperplane $x_1+x_2=0$, and
\[
\rnnl_1(T_{d,n})
=\frac12-\frac{1}{2^{n-1}}\binom{n-1}{(n-1)/2}.
\]
On $x_1=x_2$, parametrized by $(x_2,\ldots,x_n)$, the linear function $x_2$ is a closest affine function. Thus~\textup{\cite[Conjecture~2]{SergeAnaPante2025}} holds.
\end{corollary}

\begin{proof}
The value and hyperplane follow from Theorem~\textup{\ref{thm:threshold-main}}. In the hyperplane restriction formula, the ambient pair $e_1,e_2$ corresponds to the restricted character $x_2$. Both ambient coefficients are $L=2\binom{n-1}{(n-1)/2}$, so the restricted unnormalized coefficient at this character is $(L+L)/2=L>0$. The restriction formula and Proposition~\textup{\ref{prop:threshold-B12}} bound every restricted unnormalized coefficient by $L$, so this coefficient has maximal magnitude. For a Boolean function $g$ in $n-1$ variables, the normalized distance to the linear function $\ell=x_2$ is $\frac12(1-2^{-(n-1)}W_g(\ell))$. Here it is $\frac12(1-2^{-(n-1)}L)=\frac12-2^{-(n-1)}\binom{n-1}{(n-1)/2}$. Since the positive coefficient at $\ell$ attains the largest absolute coefficient, this distance equals $\nnl(g)$, proving that $x_2$ is a closest affine function.
\end{proof}

Finally, Theorems~\textup{\ref{thm:Pk-rnnl-limit}} and~\textup{\ref{thm:threshold-main}} give an exact asymptotic formula for the parameter studied in this paper.
\begin{corollary}\label{cor:threshold-Pk-limit}
For all integers $n\ge3$ and $1\le d\le n$,
\[
\lim_{\substack{k\to\infty\\ k\text{ odd}}}
\bigl(1-2P_{k,1}(T_{d,n})\bigr)^{1/(k+1)}
=1-2\nnl(T_{d,n})+\frac{1}{2^{n-1}}\binom{n-1}{d-1}.
\]
\end{corollary}

\begin{proof}
Apply Theorem~\textup{\ref{thm:Pk-rnnl-limit}} with $t=1$ and substitute the value of $\rnnl_1(T_{d,n})$ from Theorem~\textup{\ref{thm:threshold-main}} and Remark~\textup{\ref{rem:threshold-endpoints}}.
\end{proof}

\section{Conclusion}\label{sec:conclusion}

For odd $k\ge3$ and $0\le t<n$, we determined $P_{k,t}$ for direct sums of $m$ monomials of degrees at least two. When $t<m$, a minimizing restriction sets one variable equal to zero in each of the $t$ monomials of lowest degree. When $m\le t<n$, the minimum is zero. The proof combines factorization of even Walsh moments with contraction under the affine coupling produced by restriction.

The general affine-restriction identity gives a spectral characterization of $\rnnl_t(f)$. For hyperplanes, the minimum depends on the two largest absolute normalized Walsh coefficients, counted at distinct frequencies. We proved that $(1-2P_{k,t}(f))^{1/(k+1)}$ converges to $1-2\rnnl_t(f)$ as odd $k$ tends to infinity. The exact threshold formula therefore determines this limit for $t=1$ and $f=T_{d,n}$. The threshold formula also proves both conjectures in~\cite{SergeAnaPante2025}, including the odd-majority case, where $x_1+x_2=0$ is a minimizing hyperplane.

Finally, the hyperplane-moment identity shows that every non-affine function has a hyperplane restriction with strictly smaller odd-order nonhomomorphicity. The sequence $P_{k,t}(f)$ consequently decreases strictly until it reaches zero.

\end{document}